\documentclass[conference]{IEEEtran}
\IEEEoverridecommandlockouts

\usepackage{amsmath,amssymb,amsthm}
\usepackage{stmaryrd}
\usepackage{booktabs}
\usepackage{enumitem}
\usepackage{microtype}
\usepackage{xcolor}
\usepackage{tikz}
\usetikzlibrary{decorations.pathreplacing}
\definecolor{accent}{RGB}{31,74,125}
\usepackage[colorlinks,linkcolor=accent,citecolor=accent,urlcolor=accent]{hyperref}

\newtheorem{theorem}{Theorem}
\newtheorem{proposition}{Proposition}
\newtheorem{lemma}{Lemma}
\newtheorem{corollary}{Corollary}
\theoremstyle{definition}
\newtheorem{definition}{Definition}
\theoremstyle{remark}
\newtheorem{remark}{Remark}
\newtheorem{example}{Example}

\usepackage{cleveref}
\crefname{proposition}{Proposition}{Propositions}
\Crefname{proposition}{Proposition}{Propositions}
\crefname{lemma}{Lemma}{Lemmas}
\Crefname{lemma}{Lemma}{Lemmas}
\crefname{corollary}{Corollary}{Corollaries}
\Crefname{corollary}{Corollary}{Corollaries}
\crefname{definition}{Definition}{Definitions}
\Crefname{definition}{Definition}{Definitions}
\crefname{remark}{Remark}{Remarks}
\Crefname{remark}{Remark}{Remarks}
\crefname{example}{Example}{Examples}
\Crefname{example}{Example}{Examples}
\crefname{section}{Section}{Sections}
\Crefname{section}{Section}{Sections}

\newcommand{\Var}{\mathit{Var}}
\newcommand{\fv}{\mathit{fv}}
\newcommand{\asgn}{\mathit{asgn}}
\newcommand{\LoI}{\mathrm{LoI}}
\newcommand{\QoI}{\mathrm{QoI}}
\newcommand{\tc}{\mathrm{tc}}
\newcommand{\mix}{\mathrm{mix}}
\newcommand{\supp}{\mathrm{supp}}
\newcommand{\ord}{\mathrm{ord}}

\newcommand{\Ch}{\mathrm{Ch}}
\newcommand{\Uset}{\mathcal{U}}
\newcommand{\Umul}{\mathcal{U}_{\mathbb{N}}}
\newcommand{\Uloc}{\mathcal{L}}
\newcommand{\dc}{\mathord{\downarrow}}
\newcommand{\ten}{\otimes}
\newcommand{\lub}{\bigvee}
\newcommand{\derM}[4]{#1 \vdash_{M} #2\;\{#3\}\;#4}

\newcommand{\aln}{\alpha^{\mathbb{N}}}
\newcommand{\Alg}{\mathcal{A}}
\newcommand{\Jset}{\mathbb{J}}
\newcommand{\sem}[1]{\llbracket #1 \rrbracket}
\newcommand{\imp}{\Rightarrow}
\newcommand{\irule}[2]{\ensuremath{\dfrac{\;#1\;}{\;#2\;}}}

\begin{document}

\title{Renaming or Tightness:\\Enforcing Disjunctive Information Flow Policies}

\author{%
\IEEEauthorblockN{Xin Xu}
\IEEEauthorblockA{\textit{Carnegie Mellon University}\\
xuxin@cmu.edu}
\and
\IEEEauthorblockN{Siru Tao}
\IEEEauthorblockA{\textit{Carnegie Mellon University}\\
sirutao@andrew.cmu.edu}
\and
\IEEEauthorblockN{Kaizhen Tan}
\IEEEauthorblockA{\textit{Carnegie Mellon University}\\
kaizhent@cmu.edu}
}

\maketitle

\begin{abstract}
A disjunctive policy allows a value to depend on at most one of two secrets
and never on both: an analyst may consult one client's file or the other's, a
share of a split secret may be released but not its sibling. Such policies
are not lattice-shaped, and Hunt and Sands introduced the quantale of
information to give them a semantics, leaving the enforcement layer open. We
build the flow-sensitive type system family that the quantale calls for, and
show that the object which makes such families useful, the universal type
object from which every member specialises, splits in two, with a
consequence for enforcement. Over the free commutative quantale on the
program variables the whole mechanism survives for every policy: monotone
renaming, canonical derivations, principal typings, internal completeness.
Over the free object with idempotent generators the certified bound is
strictly more precise and still sound, because it records that two reads of
one source honour one disjunct. The gap cannot be closed from inside the independent-attribute family: no mechanism of that shape whose labelling maps support monotone renaming certifies a bound more precise than the first, and for principal typings under generator-exact homomorphic specialisation the two coincide. Under the ethical-wall and secret-sharing labels the second read of a disjunctive source therefore drives every such certificate to \emph{no guarantee}, and programs that satisfy the policy are rejected. The literal transcription of the lattice-era object is no escape
either: it is a further quotient that loses branch disjunction. Precision is
recovered by deferring specialisation to the judgement level, and the resulting read-out map is the least sound join-preserving one.
\end{abstract}

\section{Introduction}

A consultancy holds the files of two competing clients. Its ethical wall, the Chinese Wall of Brewer and Nash~\cite{BrewerNash89}, says that no report may depend on both. A key-management service splits a secret into two shares~\cite{Shamir79} and releases at most one. A clinical dashboard may show a
patient's diagnosis or their identity, never the pair. What these policies
forbid is not a release but a \emph{combination} of otherwise permissible
releases, and no lattice of security levels expresses that: in a lattice the
join of two admissible dependencies is admissible again.

Hunt and Sands answered this with the \emph{quantale of
information}~\cite{HuntSands21}, which keeps the lattice of
information~\cite{LandauerRedmond93} as a special case and adds a second
combination operator. Disjunction remains the join, and conjunction becomes a
tensor $\otimes$ that is \emph{not} idempotent, so that combining a
disjunctive dependency with itself may reveal strictly more than the
dependency alone. That is exactly what an ethical wall is about.

Lattice-based enforcement is a mature line, from certification through security type systems to flow-sensitive analyses~\cite{DenningDenning77,Volpano96,SabelfeldMyers03,HuntSands06}. What has been missing for the quantale is enforcement. \cite{HuntSands21} closes by naming the
step: adapt the flow-sensitive dependency analysis of~\cite{HuntSands06} so
that it becomes the language of types over the quantale. Disjunctive policies
have since been enforced for a database query language against a
purpose-built condition~\cite{Ahmadian24}, but the general question of what a
type system can certify about a disjunctive policy, and of what it must fail
to certify, has not been asked.

The analysis they point to is not an arbitrary starting point. Its
distinguishing feature is a \emph{universal object}. Fixing the flow lattice
to be the powerset of program variables yields one system that subsumes the
whole family: each program has a principal typing there, every typing under
every other lattice follows by monotone renaming, and no member of the family
certifies more than the target lattice's own system does. For enforcement
this is the difference between analysing a program once and re-analysing it
for every policy.

The suggestion even names an object: the abstract domain
of~\cite[\S V]{HuntSands21}, sets of sets of variables kept irredundant under
inclusion, which is the set object $\Uset$ below presented by antichains.
This paper carries out the step and reports what follows. The universal
object does not break in the passage to a quantale. It splits, the object
their roadmap names is one half of the split, and which half one takes
decides what can be enforced.

\subsection*{One program, two answers}

\begin{example}[the second read]\label{ex:running}
Client files $r_1$ and $r_2$ are independent. An upstream process produces a
summary $s$ under the ethical wall, so in any run $s$ reveals at most $r_1$
or at most $r_2$: its label is the disjunctive element
$\mathbb{S} = \tc\{\sim_1,\sim_2\}$ of the quantale. A report generator uses
the summary twice, once in a header and once in a table:
\[
C_3 \;=\; (u := s;\; w := s).
\]
Is the pair $(u,w)$ within the wall? It is. Whichever client's data the
summary honoured on this run, both copies honour the same one, so an observer
of $u$ and $w$ learns exactly what an observer of $s$ learns. But a type system
whose environments hold elements of the quantale must give $u$ and $w$ the
type $\mathbb{S}$ each, and the joint observation of two variables is their
tensor: $\mathbb{S} \ten \mathbb{S} = \top$. The certificate says the report
may reveal both clients' files. The mechanism has forgotten that the two
copies are one choice.
\end{example}

The forgetting is not an artefact of a rule. It is forced, and its scope is
exactly the family of programs the quantale was introduced for.

\subsection*{The universal object splits}

The powerset of variables carries two structures at once, and a lattice
cannot tell them apart. Read as an ordered monoid under union it is the free
commutative \emph{idempotent} monoid on the variables. With no idempotence
imposed, the free commutative monoid has multisets as its elements. Freeness
of the associated quantale asks a labelling to extend to a monotone monoid
homomorphism, so the idempotent reading forces every label to satisfy
$\sigma(v)\ten\sigma(v) = \sigma(v)$ while the multiset reading forces
nothing beyond affineness. Every labelling into a locale is idempotent, so in
the lattice setting the two objects specialise identically and nothing
distinguishes them. Over a quantale they come apart, and each keeps one half
of what the lattice case gave for free:

\begin{itemize}[leftmargin=1.2em]
\item Over the multiset object $\Umul$ the whole of~\cite{HuntSands06}
survives \emph{unconditionally}: monotone renaming, canonical derivations,
principal typings and internal completeness, for every affine labelling
(\Cref{thm:survive}). The $M$-indexed system is exactly the image of $\Umul$
under specialisation (\Cref{thm:exact-image}). This object counts uses, so on
\Cref{ex:running} it answers $\top$.
\item Over the set object $\Uset$ specialisation is no longer a homomorphism.
It factors as a strict homomorphism after an oplax section that crushes
multiplicities (\Cref{prop:decomp}), and the bound it produces is smaller. It
is sound for every labelling (\Cref{thm:psnd}) and strictly smaller as soon
as a disjunctive source is read twice (\Cref{thm:t3}). On
\Cref{ex:running} it answers $\mathbb{S}$, which is the truth.
\end{itemize}

A third reading suggests itself: keep the object of~\cite{HuntSands06}
literally, sets of variables with a single combination operator. That is the
free \emph{locale} on the variables, a further quotient of $\Uset$. It is
tight on \Cref{ex:running}, since union is idempotent, but it identifies the
join with the tensor and so erases the difference between ``this run depended
on $X_1$, that run on $X_2$'' and ``some run depended on $X_1 \cup X_2$'',
which is precisely the distinction the quantale exists to make
(\Cref{cor:loc}). Of the three objects, the multiset one loses same-source
correlation, the locale one loses branch disjunction, and only $\Uset$ keeps
both.

\subsection*{Renaming or tightness, not both}

Two escapes suggest themselves: write better rules, or build a better
universal object. Neither is open. Any sound mechanism that gives each
variable an element of $M$ and reads an observation off as a tensor
over-counts \Cref{ex:running}, by an argument that names no rule
(\Cref{prop:ceiling}). And a universal object is precision-neutral: if its
specialisation is monotone, lax monoidal and lax unital, which is what
transports derivations, then it certifies nothing more precise than re-analysing the program directly in the target quantale (\Cref{thm:nogo}), and for principal typings under generator-exact homomorphic specialisation the two bounds coincide (\Cref{prop:exact}). Reuse is what such an
object buys, and reuse is all it buys:

\begin{quote}
\emph{On a quantale one may have renaming or tightness, but not both.}
\end{quote}

In enforcement terms the consequence is sharp. Call a label
\emph{saturated} when $\sigma(v)\ten\sigma(v) = \top$, which is the shape of
the ethical-wall label of \Cref{ex:running} and of a two-share split, since
there the two permitted dependencies jointly determine the secret. Under a
saturated label, a single read of a source can be certified against a
non-trivial policy and a second read certifies nothing at all
(\Cref{cor:sat}), and there are programs that satisfy the policy, are
accepted by the mechanism of \Cref{sec:repair}, and are rejected by every
mechanism in the renaming paradigm (\Cref{thm:fr}). \Cref{ex:running} is one
of them. The loss saturates: by \Cref{prop:order} it is paid in full at the
second read and later reads are free.

Escaping costs exactly the renaming property. Typing in $\Uset$ and applying
the labelling to \emph{judgements} rather than to environments recovers the
tight bound (\Cref{thm:t4}), and that read-out map is the least sound join-preserving one (\Cref{thm:x}), so the repair does not merely help, it lands on the floor
while the renaming paradigm cannot reach below the ceiling.

\subsection*{Contributions}

\begin{enumerate}[leftmargin=1.4em]
\item A type system family $F_1$ over an arbitrary commutative affine
quantale, every rule anchored to a composition lemma of~\cite{HuntSands21},
sound against a mediator-based semantics of disjunctive policies (\Cref{sec:system}). Loops need the guard factor
replaced by its idempotent closure, which is invisible in a lattice
(\Cref{rem:while}).
\item The splitting of the universal object, the trichotomy with the free
locale, and the factorisation of specialisation through an oplax section
(\Cref{sec:objects}).
\item Unconditional survival of principal typings and internal completeness
over $\Umul$, and the identification of the $M$-system as exactly its image
(\Cref{sec:survive}).
\item Soundness of the set bound for arbitrary labels, the strictness
witness, the dichotomy with a local form that is checkable per program, and
the saturation bound on the loss (\Cref{sec:tight}).
\item An impossibility theorem for the renaming paradigm and its enforcement
reading: secure programs are rejected, and under saturated labels the second
read certifies nothing (\Cref{sec:nogo}).
\item The repair, minimality of its read-out, the resulting sandwich,
decidability of all three objects, a verification carried out end to end, and
a separation of the three ways a walled source gets used
(\Cref{sec:repair,sec:patterns}).
\end{enumerate}

\section{Disjunctive policies and what a run reveals}\label{sec:policies}

\subsection{Secrets, observers, policies}

Fix a finite set $D$ of secrets. $\LoI(D)$~\cite{LandauerRedmond93} is the
set of equivalence relations on $D$ ordered by $P \sqsubseteq Q$ iff
$Q \subseteq P$, so larger means more informative, with least element
$\mathrm{All}$, greatest element $\mathrm{Id}$ and join
$P \sqcup Q = P \cap Q$. An observation is a function $h$ on $D$, and
$h : P \imp \mathrm{Id}$ means $P \subseteq \ker h$: an observer of $h$
cannot separate secrets that $P$ identifies.

The quantale of information~\cite{HuntSands21} replaces one relation by a
tiling-closed set of them. For $\mathbb{P} \subseteq \LoI(D)$,
$\mix(\mathbb{P})$ collects the relations each of whose cells is a cell of
some member of $\mathbb{P}$, and $\tc(\mathbb{P}) = \dc\mix(\mathbb{P})$.
$\QoI(D)$ consists of the tiling-closed subsets ordered by inclusion, with
$\lub_i \mathbb{P}_i = \tc(\bigcup_i \mathbb{P}_i)$,
$\mathbb{P}\ten\mathbb{Q} = \tc\{P \sqcup Q \mid P\in\mathbb{P},\,
Q\in\mathbb{Q}\}$ and unit $1 = \{\mathrm{All}\}$; it is a commutative
quantale~\cite[Prop.~4]{HuntSands21}. Writing
$h : \mathbb{P}\imp\mathrm{Id}$ for ``$\exists P \in \mathbb{P}$ with
$P \subseteq \ker h$'', the two operators are disjunctive and conjunctive
combination of permissions, and weakening
($\mathbb{P}\sqsubseteq\mathbb{Q}$ and $h : \mathbb{P}\imp\mathrm{Id}$ imply
$h : \mathbb{Q}\imp\mathrm{Id}$) is~\cite[Prop.~1]{HuntSands21}. The tensor is
idempotent only on principal ideals. For the ethical-wall label
$\mathbb{S} = \tc\{\sim_1,\sim_2\}$ of \Cref{ex:running},
$\mathbb{S}\ten\mathbb{S} = \top$.

The attacker-facing reading of $\imp$ is the knowledge-based
one~\cite[\S V]{HuntSands21}. A \emph{knowledge set} of an observation $h$ is
a nonempty preimage $h^{-1}\{e\}$, what an observer of a single output can
deduce about the secret, and $h : R\imp\mathrm{Id}$ holds exactly when every
knowledge set of $h$ is tiled by $R$. A permission is thus an upper bound on
each individual deduction, not an average or an expectation, and combining
permissions with $\ten$ asks what an observer of both outputs can deduce at
once.

\begin{remark}[what tiling closure means operationally]\label{rem:perrun}
A witness $P$ for $h : \mathbb{P} \imp \mathrm{Id}$ with
$\mathbb{P} = \tc\{P_1,P_2\}$ need not be $P_1$ or $P_2$: it may be a mix,
whose cells come from $P_1$ in one region of the secret space and from $P_2$
in another. This is the operational content of a disjunctive permission, and
it is what the knowledge-based reading of the join
in~\cite[\S V]{HuntSands21} says: what an observer may deduce on a given run
is consistent with one of the disjuncts. \emph{Which} disjunct is honoured
may depend on the secret, hence on the run, and within a run exactly one is
honoured. Everything in this paper turns on that sentence: the second read of
a source in \Cref{ex:running} is safe because it happens in the same run as
the first.
\end{remark}

Throughout, $M$ denotes a commutative complete quantale
$\langle M,\sqsubseteq,\lub,\ten,1\rangle$ with $\ten$ distributing over
arbitrary joins, and $M_1 = \{m \mid 1 \sqsubseteq m\}$ its affine part;
environments take values in $M_1$, and the nonempty elements of $\QoI(D)$ are
exactly its affine elements~\cite[Lem.~3]{HuntSands21}. The variable set $\Var$ and the set $D$ of secrets are finite, and
$\mathit{Val}$ contains an encoding of every finite quotient of $D$ and
supports injective pairing, which for finite $D$ is a size assumption only.
The paper works at two levels: typing, renaming, algorithmic inference and
the universal objects are stated for an arbitrary such $M$, while every
statement that mentions mediators, compliance, soundness or tightness fixes
$M = \QoI(D)$, whose elements are the sets of equivalence relations that make
legality meaningful.

\subsection{Policies, compliance, mechanisms}

\begin{definition}[policy and compliance]\label{def:policy}
A \emph{policy} is a triple $(\sigma, W, \mathit{Pol})$ with an input
labelling $\sigma : \Var \to M_1$, an observed set $W \subseteq \Var$ and a
target $\mathit{Pol} \in M_1$. Information enters through \emph{mediators}: a
family $g = \langle g_v : D \rightharpoonup \mathit{Val}\rangle_{v\in\Var}$
with a common domain, writing $\hat g(d)$ for the store whose $v$-component
is $g_v(d)$; $g$ is \emph{legal} for $\sigma$ when $g_v : \sigma(v)\imp\mathrm{Id}$ for every $v$, that is, some member of $\sigma(v)$ identifies any two secrets that $g_v$ does not separate; for total $g_v$ this says exactly $\ker g_v\in\sigma(v)$, by downward closure. A program $C$ \emph{complies} with
the policy when
$\pi_W \circ \sem{C}\circ\hat g : \mathit{Pol} \imp \mathrm{Id}$
for every legal $g$.
\end{definition}

Mediators are what makes disjunction operational, and they are the reason
\Cref{def:policy} is not phrased as a two-run condition on stores. A
disjunctive label constrains what a value \emph{may} reveal, and to say that
two variables are served by one disjunctive source one has to name the source;
a condition quantifying over pairs of stores has no place to record it.
Several variables may therefore draw on one secret $d$, with the correlation
between them living in $D$. In \Cref{ex:running} both copies are served by
the mediator of $s$, which is why one witness certifies both. When all labels
are principal the source can be taken to be the store itself and
\Cref{rem:ni} recovers the usual two-run formulation. An \emph{enforcement mechanism} computes a bound
and accepts $C$ when the bound is below $\mathit{Pol}$. All mechanisms
considered here are sound, so acceptance implies compliance, and the question
throughout is which compliant programs are accepted.

\begin{remark}[the classical case]\label{rem:ni}
When every label and the target are principal, say
$\sigma(v) = \dc P_v$ and $\mathit{Pol} = \dc Q$, compliance says that
whenever the inputs are no more informative than the $P_v$ the observation of
$W$ is no more informative than $Q$. That is termination-insensitive
noninterference for the corresponding equivalence relations, so
\Cref{def:policy} is the standard condition with disjunctive permissions
added, not a new one.
\end{remark}

\begin{example}[the ethical wall, computed]\label{ex:wall}
Take $D = \{0,1\}^2$, a secret being a pair $(r_1,r_2)$ of independent client
files, and let $\sim_i$ be the kernel of the $i$th projection, with cells
$\{00,01\},\{10,11\}$ and $\{00,10\},\{01,11\}$. The wall permits a value to
reveal $r_1$ or to reveal $r_2$, so its label is
$\mathbb{S} = \tc\{\sim_1,\sim_2\}$. Computing it: a relation lies in
$\mix\{\sim_1,\sim_2\}$ when each of its cells is one of those four
two-element sets, and the only ways to cover $D$ by them are $[\sim_1]$ and
$[\sim_2]$, so $\mix\{\sim_1,\sim_2\} = \{\sim_1,\sim_2\}$ and
$\mathbb{S} = \{\mathrm{All},\sim_1,\sim_2\}$. The parity relation
$\ker(r_1\oplus r_2)$, with cells $\{00,11\},\{01,10\}$, is not in
$\mathbb{S}$: those cells are cells of neither $\sim_i$ and cannot be split.
Combining the wall with itself gives
$\sim_1\sqcup\sim_2 = \mathrm{Id}\in\mathbb{S}\ten\mathbb{S}$, so
$\mathbb{S}\ten\mathbb{S} = \top$. One value under the wall is harmless and
two independent values under the wall reveal everything, which is the
distinction a lattice cannot draw and the one this paper is about.
\end{example}

\begin{remark}[compound policies]\label{rem:compound}
Policies in~\cite[\S IV]{HuntSands21} are conjunctions
$\bigwedge_i (W_i, \mathit{Pol}_i)$ of requirements on several observations,
each $\mathit{Pol}_i$ a disjunction of permitted dependencies. Compliance
with a conjunction is compliance with each conjunct, and a single derivation
yields a bound for every $W$ at once (\Cref{def:joint}), so one analysis
discharges a compound policy conjunct by conjunct. Nothing below depends on
the number of conjuncts, and we work with a single one.
\end{remark}

\section{A type system family over a quantale}\label{sec:system}

\subsection{Rules}

The language is the \textbf{While} language of~\cite{HuntSands06}. Judgements
are $\derM{p}{\Gamma}{C}{\Gamma'}$ with $\Gamma : \Var \to M_1$, $p \in M_1$.
Write $t_E(\Gamma) = \bigotimes_{v\in\fv(E)}\Gamma(v)$ and
$m^{*} = \lub_{n\geq 0} m^{n}$.

\begin{figure}[t]
\centering
\small
\begin{tabular}{@{}c@{}}
\irule{}{\derM{p}{\Gamma}{\textsf{skip}}{\Gamma}} \\[2.2ex]
\irule{t = t_E(\Gamma)}
      {\derM{p}{\Gamma}{x := E}{\Gamma[x \mapsto p \ten t]}} \\[2.2ex]
\irule{\derM{p}{\Gamma}{C_1}{\Gamma''} \quad
       \derM{p}{\Gamma''}{C_2}{\Gamma'}}
      {\derM{p}{\Gamma}{C_1; C_2}{\Gamma'}} \\[2.2ex]
\irule{t = t_E(\Gamma) \quad \derM{p \ten t}{\Gamma}{C_i}{\Gamma_i'}}
      {\derM{p}{\Gamma}{\textsf{if}\;E\;C_1\;C_2}{\Gamma'}} \\[1.0ex]
\multicolumn{1}{@{}c@{}}{\footnotesize
 $\Gamma'(x) = t \ten (\Gamma_1'(x) \vee \Gamma_2'(x))$ for
 $x \in \asgn(C_1) \cup \asgn(C_2)$, else $\Gamma(x)$} \\[2.2ex]
\irule{t = t_E(\Gamma) \quad \derM{p \ten t}{\Gamma}{C}{\Gamma}}
      {\derM{p}{\Gamma}{\textsf{while}\;E\;C}{\Gamma^{\dagger}}} \\[1.0ex]
\multicolumn{1}{@{}c@{}}{\footnotesize
 $\Gamma^{\dagger}(x) = t^{*} \ten \Gamma(x)$ for $x \in \asgn(C)$,
 else $\Gamma(x)$} \\[2.2ex]
\irule{\derM{p}{\Gamma}{C}{\Gamma_1'} \quad
       p' \sqsubseteq p \quad \Gamma' \sqsubseteq \Gamma \quad
       \Gamma_1' \sqsubseteq \Gamma_2'}
      {\derM{p'}{\Gamma'}{C}{\Gamma_2'}}
\end{tabular}
\caption{The family $F_1$, indexed by a commutative affine quantale $M$.}
\label{fig:rules}
\end{figure}

Each rule transcribes an already proved composition lemma
of~\cite{HuntSands21}: their Lemma~5 for \textsc{Seq}, their Lemma~6, whose
premise has the shape $\mathbb{P}\ten(\mathbb{R}_1\vee\mathbb{R}_2)$, for
\textsc{If}, and their Lemma~8 for multi-variable expressions in
\textsc{Assign} and for the joint observation below. \Cref{app:anchor} gives
the correspondence rule by rule and argues that it is the faithful reading.

\begin{proposition}[the lattice case sits inside]\label{prop:conserv}
The idempotent affine elements of $\QoI(D)$ are the principal ideals, and
they form a lattice isomorphic to $\LoI(D)$ in which $\ten$ is the
join~\cite[Prop.~5]{HuntSands21}. On environments valued there,
\textnormal{\textsc{Assign}}, \textnormal{\textsc{Seq}},
\textnormal{\textsc{While}} and \textnormal{\textsc{Sub}} of
\Cref{fig:rules} are literally the rules of~\cite{HuntSands06} over that
lattice, since $p\ten t = p\sqcup t$, $t^{*} = t$ and the guard factor of a
loop is absorbed by an assigned variable's type. Hence on conditional-free
programs the two systems derive the same judgements.
\end{proposition}

The one rule that is not a transcription is \textsc{If}, and only because the
principal ideals are not closed under $\vee$: where~\cite{HuntSands06} must
join the two branch types, \Cref{fig:rules} keeps their disjunction, which is
the expressiveness $\QoI$ was introduced for and which \Cref{cor:loc} shows
is not free to give up.

\begin{definition}[joint observation]\label{def:joint}
A derivation $\derM{p}{\Gamma}{C}{\Gamma'}$ certifies, for every
$W \subseteq \Var$, the bound $\bigotimes_{w\in W}\Gamma'(w)$ on what an
observer of $W$ learns in the final state.
\end{definition}

\begin{remark}[the guard factor of a loop]\label{rem:while}
In~\cite{HuntSands06} a loop's postcondition is the invariant itself, a form
that relies on guard information being absorbed by an idempotent join. Over a
quantale, certifying ``the loop ran exactly $k$ times'' consumes one member
of $t$ per iteration and these members may differ, so the postcondition
carries the idempotent closure $t^{*}$, the least idempotent affine element
above $t$. When $\otimes$ is idempotent, that is when $M$ is a locale,
$t^{*} = t$ and the rule of~\cite{HuntSands06} is recovered verbatim.
With the failed free step of monotone renaming, the tensor in the joint
observation and the guard factor in \textsc{If}, this is the fourth place
where the lattice assumption was doing silent work.
\end{remark}

\subsection{Why these rules}\label{app:anchor}
\begin{table}[h]
\centering
\small
\begin{tabular}{@{}ll@{}}
\toprule
Rule & Anchor \\
\midrule
\textsc{Skip}, \textsc{Sub} & \cite{HuntSands06}, verbatim \\
\textsc{Assign} & \cite[Lem.~8]{HuntSands21} (pairing) and
  \cite[Prop.~5]{HuntSands21} \\
\textsc{Seq} & \cite[Lem.~5]{HuntSands21} (chain rule) \\
\textsc{If} & \cite[Lem.~6]{HuntSands21}, premise shape
  $\mathbb{P}\ten(\mathbb{R}_1\vee\mathbb{R}_2)$ \\
\textsc{While} & \cite{HuntSands06} plus \Cref{rem:while} \\
joint observation & \cite[Lem.~8]{HuntSands21} \\
\bottomrule
\end{tabular}
\end{table}

Three points make this the faithful reading rather than one of many. First,
the assignment of operators is fixed by the source: \cite{HuntSands21}
defines $\ten$ as conjunctive and $\vee$ as disjunctive combination, so a
multi-variable expression, whose value determines all of its operands'
contributions at once, must combine with $\ten$, and a branch merge, where a
run follows one side or the other, must combine with $\vee$. Second, the
shapes are theirs, not ours: the \textsc{If} rule is the premise of their
Lemma~6 transcribed with environments in place of single functions, and the
joint observation of a set of variables is their Lemma~8, the only pairing
rule the source provides. Third, \cite[Prop.~5]{HuntSands21} states that the embedding of $\LoI$ into
$\QoI$ sends the lattice join to $\ten$, so on environments valued in
principal ideals every rule above computes exactly what the corresponding
rule of~\cite{HuntSands06} computes over $\LoI$: the tensor is the lattice
join, and the guard factors are absorbed because an assigned variable's type
already dominates the context, with $t^{*} = t$ by idempotence. On
straight-line code F1 therefore reproduces~\cite{HuntSands06} verbatim. The
principal ideals are not closed under $\vee$, so a branch merge leaves them,
and there F1 records the disjunction that~\cite{HuntSands06} must join away.
That difference is not a liberty we take with their rules but the content of
\Cref{cor:loc}, and it is what the quantale was introduced to express.

The one place where a choice remains is the guard factor of \textsc{While}
(\Cref{rem:while}), and it is forced by soundness rather than by fidelity.
None of this constrains \Cref{sec:nogo}: \Cref{prop:ceiling} names no rule at all, and \Cref{thm:nogo} quantifies over every renaming-stable rule set, so a reader who prefers different rules inherits the same ceiling.

\subsection{Soundness}

\begin{definition}[semantic judgement]\label{def:sem}
$\models p\;\Gamma\,\{C\}\,\Gamma'$ holds when (i) every $x$ with
$p \not\sqsubseteq \Gamma'(x)$ is unchanged by $C$, and (ii) for every legal
$g$ and every $W$, $\pi_W \circ \sem{C}\circ\hat g :
\bigotimes_{w\in W}\Gamma'(w) \imp \mathrm{Id}$.
\end{definition}

\begin{remark}[relation to the source semantics]\label{rem:defs}
\Cref{def:sem} is~\cite[Def.~3]{HuntSands21} lifted from a single function to
an environment and relaxed to partial functions. The lifting is what lets a
judgement speak about every observed set at once, and the relaxation is the
usual termination-insensitive reading: $\imp$ constrains a pair of secrets
only where both runs converge, while a witness is still an equivalence
relation on all of $D$. Nothing else changes, and the composition,
conditional and pairing lemmas of the source are used as given.
\end{remark}

\begin{theorem}[soundness]\label{thm:t0}
$\derM{p}{\Gamma}{C}{\Gamma'}$ implies $\models p\;\Gamma\,\{C\}\,\Gamma'$.
Consequently a derivation with $\bigotimes_{w\in W}\Gamma'(w) \sqsubseteq
\mathit{Pol}$ establishes compliance with $(\sigma,W,\mathit{Pol})$ whenever
$\sigma \sqsubseteq \Gamma$.
\end{theorem}

The induction is in Appendix~\ref{app:t0}. Sequential composition is literally
composition of mediators. The conditional is proved from a partial-function
form of~\cite[Lem.~6]{HuntSands21}, restated and proved in the $\imp$ style
so that the appendix is self-contained. The loop assembles a witness by tiling the space of secrets by iteration count, which is
\Cref{rem:perrun} used as a proof device.

\subsection{Specialisation and renaming}

The engine of~\cite{HuntSands06} moves derivations along a map between flow
lattices. Over a quantale the requirement on that map sharpens.

\begin{lemma}[renaming]\label{lem:qren}
Let $f : M_1 \to M'_1$ be monotone, lax monoidal
($f(a)\ten f(b) \sqsubseteq f(a\ten b)$) and lax unital
($1 \sqsubseteq f(1)$). Then $\derM{p}{\Gamma}{C}{\Gamma'}$ implies
$f(p) \vdash_{M'} f\circ\Gamma\,\{C\}\,f\circ\Gamma'$.
\end{lemma}
\begin{proof}
Induction on the derivation. For \textsc{Assign} the
target rule produces $f(p)\ten\bigotimes_v f(\Gamma(v))$, and iterated
laxity bounds this by $f(p\ten t)$, after which \textsc{Sub} weakens. For
\textsc{If}, laxity gives $f(p)\ten t'\sqsubseteq f(p\ten t)$ with
$t' = \bigotimes_v f(\Gamma(v))$, so \textsc{Sub} moves the hypotheses to the
context the target rule needs, and on the merge
$t'\ten(f\Gamma_1'\vee f\Gamma_2')(x)\sqsubseteq
f\bigl(t\ten(\Gamma_1'\vee\Gamma_2')(x)\bigr)$ by monotonicity and laxity.
For \textsc{While}, $f(t)^{*}\sqsubseteq f(t^{*})$ because
$f(t)^{n}\sqsubseteq f(t^{n})\sqsubseteq f(t^{*})$, and the invariant is
preserved by monotonicity. \textsc{Seq} and \textsc{Sub} are immediate.
\end{proof}

Monotonicity alone sufficed in the lattice setting because the single
combination operator was a join, for which
$\bigsqcup_v f(\Gamma(v)) \sqsubseteq f(\bigsqcup_v\Gamma(v))$ holds for
free. A tensor is extra structure over the order and gives neither
inequality, so laxity has to be assumed. \Cref{thm:t3} shows it cannot be
dropped. \Cref{lem:qren} is what the rest of the paper means by
``specialisation supports renaming''.

\subsection{The algorithmic system}\label{sec:alg}

Folding \textsc{Sub} into the other rules gives a syntax-directed system
computing a least postcondition $\Alg^{C}_M(p,\Gamma)$, namely $\Gamma$ for
$\textsf{skip}$; $\Gamma[x\mapsto p\ten t_E(\Gamma)]$ for $x := E$;
$\Alg^{C_2}(p,\Alg^{C_1}(p,\Gamma))$ for $C_1;C_2$; for
$\textsf{if}\;E\;C_1\;C_2$ the environment sending assigned $x$ to
$t\ten(\Gamma_1'(x)\vee\Gamma_2'(x))$, where $t = t_E(\Gamma)$ and
$\Gamma_i' = \Alg^{C_i}(p\ten t,\Gamma)$, and every other $x$ to $\Gamma(x)$;
and for $\textsf{while}\;E\;C$ the environment $\Gamma^{*\dagger}$ where
\[
\Gamma^{*} \;=\; \mathrm{lfp}\,\bigl(\Theta\mapsto\Gamma\vee
  \Alg^{C}(p\ten t_E(\Theta),\Theta)\bigr)
\]
and $\dagger$ is as in \Cref{fig:rules} with $t = t_E(\Gamma^{*})$.

\begin{proposition}\label{prop:a1}
For every commutative complete quantale $M$, $\Alg^{C}_M(p,\Gamma)$ is
monotone in $(p,\Gamma)$, derivable, and below every derivable postcondition.
If $f$ is a quantale homomorphism then $f\circ\Alg^{C}_{M} =
\Alg^{C}_{M'}\circ f$; if $f$ satisfies the premises of \Cref{lem:qren} then
$\Alg^{C}_{M'}(fp, f\circ\Gamma) \sqsubseteq f\circ\Alg^{C}_{M}(p,\Gamma)$.
\end{proposition}

The loop takes its invariant as a Knaster--Tarski fixed point, so no chain
condition is needed. This matters because one of the objects below has no
ascending chain condition. Write $\Delta_C$, $\Delta^{\mathbb{N}}_C$ for the
algorithmic postconditions in the two universal objects, from
$\Delta_0(v) = \hat v$ and context $1$, and
$\Jset_W = \bigotimes_{w\in W}\Delta_C(w)$,
$\Jset^{\mathbb{N}}_W = \bigotimes_{w\in W}\Delta^{\mathbb{N}}_C(w)$.

\section{Three universal objects}\label{sec:objects}

\begin{definition}\label{def:objects}
$\Umul(\Var)$ is the set of downsets of $(\mathbb{N}^{\Var},\leq)$ with
$\lub = \bigcup$, $\mathbb{M}\ten\mathbb{N} = \dc\{m+n\}$, $1 = \dc 0$;
$\Uset(\Var)$ is the set of downsets of $(\mathcal{P}(\Var),\subseteq)$ with
$\lub = \bigcup$, $\mathbb{X}\ten\mathbb{Y} = \dc\{X\cup Y\}$,
$1 = \dc\emptyset$; $\Uloc(\Var)$ is $\mathcal{P}(\Var)$ with
$\lub = \bigcup$ and $\ten = \cup$, the free locale on $\Var$ and the type
object of~\cite{HuntSands06}. Generators are $\hat v = \dc e_v$,
$\dc\{v\}$, $\{v\}$.
\end{definition}

The first two are instances of the standard downset construction on an
ordered monoid, with tensor given by Day convolution. Free quantales in this
sense are textbook material~\cite{Rosenthal90}. What matters here is the side
condition each imposes.

\begin{proposition}[universal properties]\label{prop:t1}
Let $M$ be a commutative quantale.
\begin{enumerate}[label=(\alph*),leftmargin=1.4em]
\item For every $\sigma : \Var \to M_1$ there is a unique join, tensor and
unit preserving $\aln_\sigma : \Umul \to M$ with $\aln_\sigma(\hat v) =
\sigma(v)$, namely $\aln_\sigma(\mathbb{M}) = \lub_{m\in\mathbb{M}}
\bigotimes_v \sigma(v)^{m(v)}$. Affineness is the only requirement.
\item Such an $\alpha_\sigma : \Uset \to M$ exists iff every $\sigma(v)$ is
$\ten$-idempotent, and then $\alpha_\sigma(\mathbb{X}) =
\lub_{X\in\mathbb{X}}\bigotimes_{v\in X}\sigma(v)$.
\end{enumerate}
\end{proposition}
\begin{proof}
A monotone monoid homomorphism from
$(\mathbb{N}^{\Var},+,\leq)$ to $(M,\ten,\sqsubseteq)$ is determined by the
images $\sigma(v)$ of the generators, and monotonicity is equivalent to
$1\sqsubseteq\sigma(v)$ for all $v$ because $m\leq m'$ implies
$h(m') = h(m)\ten h(m'-m)$. The downset quantale with Day convolution is the
free quantale on that ordered monoid~\cite{Rosenthal90}, which gives (a);
uniqueness also follows directly from $\mathbb{M} =
\lub_{m\in\mathbb{M}}\dc m$ and $\dc m = \bigotimes_v\hat v^{m(v)}$. For (b)
the generating monoid is $(\mathcal{P}(\Var),\cup)$, in which
$X\cup X = X$, so any monoid homomorphism has idempotent images; conversely
if all $\sigma(v)$ are idempotent then in
$\bigotimes_{v\in X\cup Y}\sigma(v)$ the repeated factors coming from
$X\cap Y$ can be absorbed, so the displayed formula preserves $\ten$, and it
preserves joins and the unit by construction.
\end{proof}

Freeness asks for a monotone monoid homomorphism out of the generating
ordered monoid. Monotonicity forces affineness in both cases, and idempotence of
$\cup$ forces idempotence of the images in the second. In a lattice the
second condition is vacuous, which is why one object could play both roles.
The formula for $\alpha_\sigma$ makes sense for every affine $\sigma$. What
idempotence buys is that it preserves the tensor.

\begin{lemma}[the two quotients]\label{lem:q}
$q(\mathbb{M}) = \{\supp m \mid m \in \mathbb{M}\}$ and
$u(\mathbb{X}) = \bigcup_{X\in\mathbb{X}} X$ are surjective quantale
homomorphisms $\Umul \to \Uset \to \Uloc$. The first has a left adjoint
$s(\mathbb{X}) = \dc\{\chi_X \mid X\in\mathbb{X}\}$ with $q\circ s =
\mathrm{id}$ and $s\circ q \sqsubseteq \mathrm{id}$, and a right adjoint
$r(\mathbb{X}) = \{m \mid \supp m \in \mathbb{X}\}$; $s$ is oplax and is not
a homomorphism.
\end{lemma}
\begin{proof}
$q$ and $u$ preserve unions by construction;
$\supp(m+n) = \supp m\cup\supp n$ and $\bigcup(\mathbb{X}\ten\mathbb{Y}) =
u\mathbb{X}\cup u\mathbb{Y}$ give the tensor, and the units match. For $q$,
note first that its image is already downward closed, since $m\in\mathbb{M}$
and $X\subseteq\supp m$ give $m|_X\leq m$ with support $X$. The adjunction
$s\mathbb{X}\subseteq\mathbb{M}$ iff $\mathbb{X}\subseteq q\mathbb{M}$ holds
because $\chi_X\leq m$ whenever $\supp m = X$; $q\circ s = \mathrm{id}$ and
$s\circ q\sqsubseteq\mathrm{id}$ follow, and $r$ is the right adjoint since
$q$ preserves joins. Finally $s(\hat v\ten\hat v) = s(\hat v)\subsetneq
\dc(2e_v) = s(\hat v)\ten s(\hat v)$.
\end{proof}

Reading the chain from the left, $q$ forgets multiplicity, that is
same-source correlation, and $u$ forgets which branch a dependency came from,
that is disjunction.

\begin{remark}[in the vocabulary of abstract interpretation]\label{rem:disjcomp}
The passage from $\Uloc$ to $\Uset$ is the disjunctive completion of the
lattice-era domain~\cite{CousotCousot79,FileRanzato99}: downsets of
$\mathcal{P}(\Var)$ are exactly the disjunctive completion of
$\mathcal{P}(\Var)$. What makes the completion pay here is on the concrete
side. \cite[\S V]{HuntSands21} points out that completing the domain achieves
nothing if a disjunction of abstract points is read as a union of concrete
ones, because the kernel of a conditional lies in neither branch's set; the
join of $\QoI$ is tiling closure, and a mix of the two branch relations is in
it. The completion and the tiling closure have to arrive together. The
passage from $\Umul$ to $\Uset$ is of a different nature: it is not a
completion but a quotient, and it is what \Cref{thm:nogo} shows no
specialisation can perform on its own.
\end{remark}

\begin{figure}[t]
\centering
\begin{tikzpicture}[>=stealth,font=\small]
  \node (un) at (0,1.35) {$\Umul$};
  \node (us) at (2.5,1.35) {$\Uset$};
  \node (ul) at (5.0,1.35) {$\Uloc$};
  \node (m)  at (2.5,-0.35) {$M$};
  \draw[->>] (un) -- node[above,font=\scriptsize]{$q$}
                    node[below,font=\scriptsize]{forgets reuse} (us);
  \draw[->>] (us) -- node[above,font=\scriptsize]{$u$}
                    node[below,font=\scriptsize]{forgets branches} (ul);
  \draw[->,bend right=35] (us) to node[above,font=\scriptsize,pos=0.5]{$s$} (un);
  \draw[->] (un) -- node[left,font=\scriptsize,pos=0.55]{$\aln_\sigma$ hom.} (m);
  \draw[->] (us) -- node[right,font=\scriptsize,pos=0.5]{$\alpha_\sigma$} (m);
  \draw[->] (ul) -- node[right,font=\scriptsize,pos=0.55]{$\lambda_\sigma$} (m);
\end{tikzpicture}
\caption{The multiset object, the set object and the free locale, with the
two quotients and the section $s$ of $q$. Only $\aln_\sigma$ is a
homomorphism for every affine $\sigma$; the other two are oplax, and
$\alpha_\sigma = \aln_\sigma\circ s$ (\Cref{prop:decomp}). The object named in
the roadmap of~\cite[\S V]{HuntSands21} is $\Uset$; the type object
of~\cite{HuntSands06} is $\Uloc$.}
\label{fig:objects}
\end{figure}

\begin{remark}[the source's own concretisation is a specialisation]
\label{rem:gammaq}
\cite[\S V]{HuntSands21} equips its abstract domain with a concretisation
into $\QoI(\mathit{Sto})$ sending a set of variable-sets to the tiling
closure of the corresponding relations. In the present notation that map is
$\alpha_{\sigma_0}$ at the all-principal labelling
$\sigma_0(v) = \dc\!\sim_v$, where $\sim_v$ is the kernel of reading $v$ from
a store. So the concretisation the source uses to justify its abstract domain
is one member of the family of specialisations studied here, and it lies on
the idempotent side of \Cref{thm:dichotomy}, where the two objects agree. The
disjunctive side is reached only by labellings the source does not
instantiate.
\end{remark}

\begin{proposition}[factorisation]\label{prop:decomp}
For every affine $\sigma$, idempotent or not,
$\alpha_\sigma = \aln_\sigma \circ s$ as maps $\Uset \to M$.
\end{proposition}
\begin{proof}
$\aln_\sigma(s\mathbb{X}) =
\lub_{X\in\mathbb{X}}\bigotimes_v\sigma(v)^{\chi_X(v)} =
\lub_{X\in\mathbb{X}}\bigotimes_{v\in X}\sigma(v) =
\alpha_\sigma(\mathbb{X})$.
\end{proof}

\Cref{prop:decomp} locates the defect precisely. Specialisation itself,
$\aln_\sigma$, is a strict homomorphism for every labelling. All of the
oplaxity comes from $s$, the step that crushes multiplicities to one. When
labels are idempotent, $\aln_\sigma$ absorbs multiplicities on its own and
crushing first changes nothing.

\subsection{The three objects on two programs}\label{sec:worked}

\begin{example}\label{ex:worked}
Alongside \Cref{ex:running}, take the branch program
\[
C_1 \;=\; \textsf{if}\;p(x)\;\textsf{then}\;w := g(y)\;
          \textsf{else}\;w := h(z)
\]
with classical labels $\sigma(x) = \dc\!\sim_x$ and so on, and observe $w$.
Universal typing gives $\Delta_{C_1}(w) = \dc\{x,y\}\cup\dc\{x,z\}$ in
$\Uset$, $\dc(2e_x{+}e_y)\cup\dc(2e_x{+}e_z)$ in $\Umul$, where the guard is
counted once for the branch it selects and once inside the branch, and
$\{x,y,z\}$ in $\Uloc$. \Cref{ex:running} instead reads one source twice
under the saturated label $\mathbb{S}$. Specialising, and noting that the
doubled guard is free here because $\sigma(x)$ is classical:
\begin{center}\small
\begin{tabular}{@{}lll@{}}
\toprule
 & $C_1$, observe $w$ & $C_3$, observe $\{u,w\}$ \\
\midrule
$\Umul$ & $\tc\{\sim_{xy},\sim_{xz}\}$ & $\top$ \\
$\Uloc$ & $\dc\!\sim_{xyz}$ & $\mathbb{S}$ \\
$\Uset$ & $\tc\{\sim_{xy},\sim_{xz}\}$ & $\mathbb{S}$ \\
\midrule
least sound bound & $\tc\{\sim_{xy},\sim_{xz}\}$ & $\mathbb{S}$ \\
\bottomrule
\end{tabular}
\end{center}
Each of the two quotients fails on exactly one program, and on a different
one. The multiset object cannot see that the two reads of $s$ are one choice.
The locale object cannot see that the two branches of $C_1$ are alternatives,
and answers as if some run had used $y$ and $z$ together. The bottom row is
attained, for $C_1$ by injective branches and a guard held constant, for
$C_3$ by \Cref{thm:t3}.
\end{example}

\section{What survives}\label{sec:survive}

One half of every Galois connection here survives a tensor unconditionally,
and it is worth isolating first, because both this section and the repair of
\Cref{sec:repair} rest on it.

\begin{lemma}[the abstraction half never fails]\label{lem:gamma}
Let $\alpha : \mathcal{U}\to M$ preserve joins, be oplax and preserve the
unit laxly, and let $\gamma$ be its right adjoint. Then $\gamma$ is monotone,
lax monoidal and lax unital, so it satisfies the premises of
\Cref{lem:qren}, whatever $\alpha$ does.
\end{lemma}
\begin{proof}
Monotonicity is standard for a right adjoint. By adjunction,
$\gamma m\ten\gamma m'\sqsubseteq\gamma(m\ten m')$ is equivalent to
$\alpha(\gamma m\ten\gamma m')\sqsubseteq m\ten m'$, and oplaxity followed by
the counit gives $\alpha(\gamma m\ten\gamma m')\sqsubseteq\alpha\gamma m\ten
\alpha\gamma m'\sqsubseteq m\ten m'$. The unit is $\alpha(1)\sqsubseteq 1$.
\end{proof}

Oplaxity of $\alpha$, the defect on the set object, is exactly what makes its
adjoint lax on the nose. The two halves of the connection part company under
a tensor, and the surviving half is the one that carries $M$-derivations back
into the universal object.

\begin{theorem}\label{thm:survive}
Over $\Umul$, for every affine $\sigma$: $\aln_\sigma$ and its right adjoint
both satisfy the premises of \Cref{lem:qren}; canonical derivations hold in
both directions; $\derM{1}{\sigma}{C}{\Gamma'}$ iff
$\aln_\sigma(\Delta^{\mathbb{N}}_C)\sqsubseteq\Gamma'$, so
$\Delta^{\mathbb{N}}_C$ is a principal typing; and no member of the family
certifies more $M$-typings than the $M$-system itself.
\end{theorem}
\begin{proof}
$\aln_\sigma$ is a homomorphism by \Cref{prop:t1}(a), hence lax, and its
right adjoint is lax by \Cref{lem:gamma}, so both transport derivations by
\Cref{lem:qren}. For canonical derivations, renaming
$\derM{1}{\sigma}{C}{\Gamma'}$ along the right adjoint and tightening with
\textsc{Sub}, using $\hat v\sqsubseteq\gamma^{\mathbb{N}}(\sigma v)$, gives a
universal derivation below $\gamma^{\mathbb{N}}\circ\Gamma'$, and conversely
renaming a universal derivation along $\aln_\sigma$ and applying the counit
gives an $M$-derivation. Principality is then \Cref{thm:exact-image} in one
direction and renaming plus \textsc{Sub} in the other, and internal
completeness is the statement that the resulting bound is attained.
\end{proof}

Every step of the lattice argument goes through with one substitution: where
the lattice proof uses that the combination operator is a join, this proof
uses that $\aln_\sigma$ is a homomorphism, which \Cref{prop:t1}(a) supplies
unconditionally. The results of~\cite{HuntSands06} are not lost in the
passage to a quantale.

\begin{theorem}[the $M$-system is an image]\label{thm:exact-image}
$\Alg^{C}_M(1,\sigma) = \aln_\sigma(\Delta^{\mathbb{N}}_C)$ pointwise, so the
greatest lower bound of the joint bounds for $W$ derivable in the $M$-system
is $\aln_\sigma(\Jset^{\mathbb{N}}_W)$.
\end{theorem}

\begin{proof}
Renaming along the homomorphism $\aln_\sigma$ makes
$\aln_\sigma(\Delta^{\mathbb{N}}_C)$ a derivable $M$-postcondition from
$\sigma$, so $\Alg^{C}_M(1,\sigma)\sqsubseteq
\aln_\sigma(\Delta^{\mathbb{N}}_C)$; the transport clause of \Cref{prop:a1}
for homomorphisms gives the converse. Every derivable postcondition dominates
$\Alg^{C}_M(1,\sigma)$, and $\aln_\sigma$ preserves $\ten$, so the greatest
lower bound of the derivable joint bounds is
$\aln_\sigma(\Jset^{\mathbb{N}}_W)$.
\end{proof}

Questions about what the $M$-system can certify therefore reduce to multiset
arithmetic: count how often each source is consumed by the observation.

\section{What it certifies}\label{sec:tight}

\subsection{The set bound is sound}

\begin{theorem}[soundness of the set bound]\label{thm:psnd}
Let $M = \QoI(D)$, $\sigma$ affine and $g$ legal for $\sigma$. Then for every
$C$ and $W$,
$\pi_W\circ\sem{C}\circ\hat g : \alpha_\sigma(\Jset_W)\imp\mathrm{Id}$.
\end{theorem}

\begin{proof}
Choose, for each $v$, a witness $c(v)\in\sigma(v)$ of $g_v$, which legality
provides, and put $\sigma_c(v) = \dc c(v)$. The family $g$ is legal for
$\sigma_c$ with the same witnesses, and every $\sigma_c(v)$ is principal,
hence $\ten$-idempotent, so $\alpha_{\sigma_c}$ is a strict homomorphism by
\Cref{prop:t1}(b) and transports the universal derivation
$1\vdash_{\Uset}\Delta_0\,\{C\}\,\Delta_C$ along \Cref{lem:qren} to
$\derM{1}{\sigma_c}{C}{\alpha_{\sigma_c}\circ\Delta_C}$. By \Cref{thm:t0},
$\pi_W\circ\sem{C}\circ\hat g :
\bigotimes_W\alpha_{\sigma_c}(\Delta_C(w))\imp\mathrm{Id}$, and that bound is
$\alpha_{\sigma_c}(\Jset_W)$ because $\alpha_{\sigma_c}$ preserves $\ten$.
Finally $\alpha_{\sigma_c}\sqsubseteq\alpha_\sigma$ pointwise, so weakening
applies.
\end{proof}

The proof is \Cref{rem:perrun} turned into a construction. A legal mediator
family \emph{is} a run's choice of disjuncts, choosing witnesses collapses
the disjunctive labelling to a classical one, and on classical labels the
lattice machinery applies unchanged. The same construction explains the
mechanism of the gap.

\begin{lemma}[choice decomposition]\label{lem:diag}
With $\Ch(\sigma) = \prod_v\sigma(v)$ and $\sigma_c(v) = \dc c(v)$,
$\alpha_\sigma = \lub_{c\in\Ch(\sigma)}\alpha_{\sigma_c}$ with each
$\alpha_{\sigma_c}$ a strict homomorphism, so
\[
\begin{aligned}
\alpha_\sigma(\mathbb{X}\ten\mathbb{Y}) &=
 \lub_{c}\bigl(\alpha_{\sigma_c}\mathbb{X}\ten
 \alpha_{\sigma_c}\mathbb{Y}\bigr),\\
\alpha_\sigma\mathbb{X}\ten\alpha_\sigma\mathbb{Y} &=
 \lub_{c,c'}\bigl(\alpha_{\sigma_c}\mathbb{X}\ten
 \alpha_{\sigma_{c'}}\mathbb{Y}\bigr).
\end{aligned}
\]
\end{lemma}

\begin{proof}
Every $\sigma(v)$ is tiling closed, so $\sigma(v) =
\lub_{P\in\sigma(v)}\dc P$. For finite $X$, distributing $\ten$ over these
joins and using $\dc P\ten\dc Q = \dc(P\sqcup Q)$ gives
$\bigotimes_{v\in X}\sigma(v) = \lub_{c}\bigotimes_{v\in X}\dc c(v)$ with $c$
ranging over choices on $X$, extended arbitrarily elsewhere. Joining over
$X\in\mathbb{X}$ and exchanging the two joins gives the decomposition. Each
$\sigma_c$ is principal, hence idempotent, so $\alpha_{\sigma_c}$ is a
homomorphism by \Cref{prop:t1}(b), and the two identities follow by
distributivity.
\end{proof}

The first is a diagonal, the second a full square. A run honours a
disjunctive permission once. The object on sets remembers that two uses were
served by one choice, and an $M$-valued environment lets every use choose
afresh.

\subsection{Strictness, dichotomy, and how much is lost}

\begin{theorem}\label{thm:t3}
Let $\sigma(v_0)$ be non-idempotent and $C_3$ as in \Cref{ex:running}. The
judgement bounding the joint observation of $\{u,w\}$ by $\sigma(v_0)$ is
semantically true and is obtained by specialising the universal typing, while
every joint bound for $\{u,w\}$ derivable in the $M$-system is above
$\sigma(v_0)\ten\sigma(v_0)$.
\end{theorem}

\begin{proof}
In $\Uset$ the algorithmic typing gives $\Delta(u) = \Delta(w) = \hat v_0$,
so $\Jset_{\{u,w\}} = \hat v_0\ten\hat v_0 = \hat v_0$, whose specialisation
is $\sigma(v_0)$; \Cref{thm:psnd} makes that bound true. In $M$ the
algorithmic typing gives $\Gamma'(u) = \Gamma'(w) = \sigma(v_0)$, every
derivable postcondition dominates it by \Cref{prop:a1}, and the joint
observation is their tensor, which is strictly above $\sigma(v_0)$ by
non-idempotence.
\end{proof}

In the set object $\hat v_0 \ten \hat v_0 = \hat v_0$ and one mediator serves
both reads, so one witness certifies both. The argument does not depend on
the rules: by the time the observation is assembled, $\Gamma'(u)$ and
$\Gamma'(w)$ are independent elements of $M$, and no rule whose environments
take values in $M$ can recover the fact that they came from one source, a
point made general in \Cref{prop:ceiling}.
\Cref{thm:t3} is also what shows that laxity cannot be dropped from
\Cref{lem:qren}: were monotonicity enough, $\alpha_\sigma$ would transport
the universal derivation and contradict the theorem.

\begin{theorem}\label{thm:dichotomy}
$\alpha_\sigma(\Jset_W) = \aln_\sigma(\Jset^{\mathbb{N}}_W)$ for every $C$
and $W$ iff every $\sigma(v)$ is $\ten$-idempotent. Over a finite set of
secrets the idempotent elements of $\QoI(D)$ are exactly the principal ideals
together with $\bot$, so the condition says that every label is a classical
equivalence-relation label.
\end{theorem}

\begin{proof}
If every $\sigma(v)$ is idempotent then $\alpha_\sigma$ is a homomorphism, so
$\alpha_\sigma(\Jset_W) = \bigotimes_W\alpha_\sigma(\Delta_C(w))$, which is
the $M$-system's least bound by \Cref{thm:exact-image} and
\Cref{prop:decomp}, the crushing having no effect. Conversely,
non-idempotence of some $\sigma(v_0)$ makes $C_3$ a counterexample by
\Cref{thm:t3}. For the characterisation, $\mathbb{P}$ is idempotent iff it is
closed under $\sqcup$: closure gives $\mathbb{P}\ten\mathbb{P}\subseteq
\tc(\mathbb{P}) = \mathbb{P}$, and the converse is $P = P\sqcup P$. A
nonempty finite $\sqcup$-closed downset has a greatest element, hence is
principal, and principal ideals are $\sqcup$-closed.
\end{proof}

The safety reading is that the two objects agree exactly on the policies a
lattice could already express, and part company on every genuinely
disjunctive one.

\begin{proposition}[the loss saturates]\label{prop:order}
Let $\ord(m)$ be the least $k\geq 1$ with $m^{k} = m^{k+1}$. Then
$\aln_\sigma(\mathbb{M}) = \lub_{m\in\mathbb{M}}\bigotimes_v
\sigma(v)^{\min(m(v),\,\ord\sigma(v))}$: over-counting a source stops
mattering after $\ord$ uses. For $\mathbb{S} = \tc\{\sim_1,\sim_2\}$ one has
$\mathbb{S}^{2} = \top = \mathbb{S}^{3}$, so the order is two.
\end{proposition}
\begin{proof}
The chain $m\sqsubseteq m^{2}\sqsubseteq\cdots$
stabilises at $\ord(m)$, so $\sigma(v)^{j} = \sigma(v)^{\min(j,\ord)}$ for
$j\geq 1$; substitute in \Cref{prop:t1}(a). For $\mathbb{S}$, the tensor
contains $\sim_1\sqcup\sim_2 = \mathrm{Id}$, so
$\mathbb{S}^{2} = \top$.
\end{proof}

\begin{proposition}[the loss is local to reused sources]\label{prop:local}
Fix $C$ and $W$. If every $v$ that occurs with multiplicity at least two in
some member of $\Jset^{\mathbb{N}}_W$ has $\sigma(v)$ idempotent, then
$\alpha_\sigma(\Jset_W) = \aln_\sigma(\Jset^{\mathbb{N}}_W)$.
\end{proposition}
\begin{proof}
For such $v$, $\sigma(v)^{k} = \sigma(v)$ for all $k\geq 1$, so
$\bigotimes_v\sigma(v)^{m(v)} = \bigotimes_{v\in\supp m}\sigma(v)$ for every
$m\in\Jset^{\mathbb{N}}_W$, and the two joins agree term by term.
\end{proof}

So the dichotomy has a checkable local form. One need not ask whether all
labels are classical, only whether the labels of the sources this particular
observation consumes more than once are, and \Cref{prop:order} weakens even
that to consuming a source more often than its order. Everything below
concerns the case where the test fails, which by \Cref{ex:wall} is the
generic case for a wall: its order is two, so a single reuse suffices.

\section{No mechanism with renaming does better}\label{sec:nogo}

Two escape routes suggest themselves for the loss in \Cref{ex:running}:
write better rules, or build a better universal object. This section closes
both. Everything below concerns mechanisms of the \emph{independent-attribute}
shape, in the terminology of~\cite[\S VIII]{HuntSands21}: environments assign
one element of $M$ to each variable, and an observation of $W$ is read out as
$\bigotimes_{w\in W}\Gamma'(w)$. \Cref{fig:rules} has that shape, and so does
every system obtained from~\cite{HuntSands06} by reading its combination
operator as the tensor.

\subsection{Rules do not help}

\begin{proposition}[the independent-attribute ceiling]\label{prop:ceiling}
Let a mechanism of the above shape be sound, by any rules whatsoever. Then
for $C_3$ of \Cref{ex:running} and $W = \{u,w\}$ it certifies a bound
$\sqsupseteq\sigma(v_0)\ten\sigma(v_0)$.
\end{proposition}
\begin{proof}
Soundness of the entry for $u$ requires
$\pi_u\circ\sem{C_3}\circ\hat g = g_{v_0}$ to satisfy $\Gamma'(u)\imp\mathrm{Id}$ for every legal $g$. Realising each $P\in\sigma(v_0)$ as the kernel of a total mediator, soundness produces a witness above $P$ in the information order, and downward closure puts $P$ itself in $\Gamma'(u)$; hence $\Gamma'(u)\sqsupseteq\sigma(v_0)$, and likewise for $w$. The read-out
is their tensor.
\end{proof}

No rule is named in that argument. Once the two entries are separate elements
of $M$, the information that they were served by one mediator is gone, and
the tensor cannot recover it.

\subsection{Objects do not help either}

\begin{definition}[the renaming paradigm]\label{def:paradigm}
Call a rule set $R$ over judgements $\derM{p}{\Gamma}{C}{\Gamma'}$
\emph{renaming-stable} if $R$-derivability transports along every map
satisfying the premises of \Cref{lem:qren}; \Cref{lem:qren} says that
\Cref{fig:rules} is one. A \emph{type object over $R$} is a commutative
complete quantale $\mathcal{V}$, an affine interpretation
$\iota:\Var\to\mathcal{V}_1$ and, for every affine $\sigma$, a map
$h_\sigma:\mathcal{V}\to M$ satisfying those premises with
$h_\sigma(\iota v)\sqsupseteq\sigma(v)$; derivations in $\mathcal{V}$ use
$R$. This is what it takes for one analysis to serve every policy by
transport of derivations.
\end{definition}

\begin{theorem}[no-go]\label{thm:nogo}
Let $R$ be renaming-stable, let $(\mathcal{V},\iota,\{h_\sigma\})$ be a type
object over $R$, and let $\Delta^{\mathcal{V}}_C$ be any postcondition
$R$-derivable in $\mathcal{V}$ from $\iota$. Then
$h_\sigma(\bigotimes_{w\in W}\Delta^{\mathcal{V}}_C(w))$ is above the
greatest lower bound of the joint bounds for $W$ that $R$ derives in $M$ from
$\sigma$ directly. For $R$ the rules of \Cref{fig:rules} that bound is $\aln_\sigma(\Jset^{\mathbb{N}}_W)$.
\end{theorem}
\begin{proof}
By renaming-stability the $\mathcal{V}$-derivation transports to
$h_\sigma(1_{\mathcal{V}})\vdash_M h_\sigma\circ\iota\,\{C\}\,
h_\sigma\circ\Delta^{\mathcal{V}}_C$, and \textsc{Sub} lowers the context to
$1$, which is below $h_\sigma(1_{\mathcal{V}})$ by lax unitality, and the
precondition to $\sigma\sqsubseteq h_\sigma\circ\iota$. So
$\bigotimes_W h_\sigma(\Delta^{\mathcal{V}}_C(w))$ is one of the joint
bounds $R$ derives in $M$ from $\sigma$, and lax monoidality puts
$h_\sigma(\bigotimes_W\Delta^{\mathcal{V}}_C(w))$ above it. For \Cref{fig:rules} the greatest lower bound is $\aln_\sigma(\Jset^{\mathbb{N}}_W)$ by \Cref{thm:exact-image}.
\end{proof}

\begin{proposition}[exactness for principal typings]\label{prop:exact}
Let the type object have every $h_\sigma$ a quantale homomorphism with
$h_\sigma(\iota v) = \sigma(v)$, and let the derivation be the algorithmic
one, $\Delta^{\mathcal{V}}_C = \Alg^{C}_{\mathcal{V}}(1,\iota)$. Then
$h_\sigma(\bigotimes_{w\in W}\Delta^{\mathcal{V}}_C(w)) =
\aln_\sigma(\Jset^{\mathbb{N}}_W)$.
\end{proposition}
\begin{proof}
Freeness (\Cref{prop:t1}(a)) gives the unique homomorphism
$k:\Umul\to\mathcal{V}$ with $k(\hat v) = \iota(v)$, and $h_\sigma\circ k$
and $\aln_\sigma$ are homomorphisms agreeing on generators, hence equal. The
transport clause of \Cref{prop:a1} gives $\Alg^{C}_{\mathcal{V}}(1,\iota) =
k(\Delta^{\mathbb{N}}_C)$, and $k$ preserves the tensor, so the certified
bound is $h_\sigma(k(\Jset^{\mathbb{N}}_W)) =
\aln_\sigma(\Jset^{\mathbb{N}}_W)$.
\end{proof}

Taking $\mathcal{V} = M$, $\iota = \sigma$ and $h_\sigma$ the identity
recovers \Cref{thm:exact-image}. A mechanism may of course certify less by
weakening. The two statements together say that the paradigm's best case is
the multiset bound and that nothing in it does better: a universal object
with a renaming-capable specialisation is \emph{precision-neutral}, and
reuse is what it buys, and all it buys. Combined with \Cref{prop:ceiling},
neither escape route is open, and two enforcement consequences follow.

\begin{corollary}[a second read certifies nothing]\label{cor:sat}
Call $m$ \emph{saturated} when $m\ten m = \top$. The ethical-wall label
$\mathbb{S}$ and the label of a two-share split are saturated, since there
the two permitted dependencies jointly determine the secret. If $\sigma(v_0)$
is saturated and some member of $\Jset^{\mathbb{N}}_W$ uses $v_0$ at least
twice, then $\aln_\sigma(\Jset^{\mathbb{N}}_W) = \top$, so by
\Cref{thm:nogo} no mechanism in the paradigm certifies $W$ against any policy
below $\top$.
\end{corollary}
\begin{proof}
If $m\in\Jset^{\mathbb{N}}_W$ has $m(v_0)\geq 2$ then
$\aln_\sigma(\Jset^{\mathbb{N}}_W)\sqsupseteq\bigotimes_v\sigma(v)^{m(v)}
\sqsupseteq\sigma(v_0)^{2} = \top$ by affineness.
\end{proof}

Under a saturated label the degradation is not gradual. One read of a source
can be certified against a tight policy. At the second read the certificate
becomes vacuous, and by \Cref{prop:order} further reads cost nothing more.

\begin{theorem}[secure programs are rejected]\label{thm:fr}
Let $\sigma(v_0)$ be saturated and consider the policy
$(\sigma,\{u,w\},\sigma(v_0))$ with $C_3$ of \Cref{ex:running}. Then $C_3$
complies with the policy and is accepted by the mechanism of
\Cref{sec:repair}, while every sound independent-attribute mechanism rejects
it, whatever its rules and whatever universal object it specialises from.
\end{theorem}
\begin{proof}
Compliance is \Cref{thm:psnd} at
$\alpha_\sigma(\Jset_{\{u,w\}}) = \sigma(v_0)$, which is also the bound
\Cref{thm:t4} computes. Rejection is \Cref{prop:ceiling} with saturation:
the certified bound is above $\sigma(v_0)^{2} = \top\not\sqsubseteq
\sigma(v_0)$.
\end{proof}

The rejected program is not a pathology constructed for the proof. It is one
source read twice: the reconstruction step of secret sharing, a sanitised
report quoted in two places, any dashboard that renders one disjunctively
labelled value in two panels. The programs the quantale of information was
introduced to reason about are exactly the ones the classical mechanism
cannot certify.

\begin{remark}[scope]\label{rem:scope}
Both results above are about independent-attribute mechanisms.
\cite[\S VIII]{HuntSands21} anticipates that under disjunctive information a
relational analysis, one able to state that \emph{outputs} $a$ and $b$ depend
on $x$ or $y$ jointly, is strictly stronger than the independent-attribute
kind, and suggests path sensitivity as a route to it, and \cite{Ahmadian24} takes
that route for database queries. Those mechanisms leave the shape assumed
here and are not covered. What \Cref{prop:ceiling} and \Cref{thm:nogo} settle
is the fate of the independent-attribute family, which is the family
in which principal typings and specialisation live, and \Cref{sec:repair}
shows that a useful part of the relational strength is available without path
sensitivity: it is enough to keep shared sources visible until the read-out.
\end{remark}

\section{Recovering precision}\label{sec:repair}

\begin{theorem}[deferred specialisation]\label{thm:t4}
Let environments take values in $\Uset$, so that typing is universal typing,
and let the labelling act only on the observation, through $\alpha_\sigma$.
The result is sound for every affine $\sigma$ (\Cref{thm:psnd}), is below
every judgement derivable with $M$-valued environments, and agrees with
$M$-valued typing iff every $\sigma(v)$ is idempotent.
\end{theorem}
\begin{proof}
Soundness is \Cref{thm:psnd}. For domination, let
$\derM{1}{\sigma}{C}{\Gamma'}$; by \Cref{lem:gamma} the right adjoint
$\gamma_\sigma$ of $\alpha_\sigma$ transports it to a universal derivation
below $\gamma_\sigma\circ\Gamma'$, minimality gives
$\Delta_C\sqsubseteq\gamma_\sigma\circ\Gamma'$, and the counit gives
$\alpha_\sigma(\Delta_C)\sqsubseteq\Gamma'$. Equality with $M$-valued typing
is \Cref{thm:dichotomy}.
\end{proof}

This is the route \Cref{thm:nogo} leaves open: $\alpha_\sigma$ is oplax
rather than lax, so it does not transport derivations and the labelling must
be applied to judgements. Renaming is given up. Tightness is bought. The next
two statements say that the trade is exactly right, first because no join-preserving read-out improves on this one, second because the other natural object is strictly coarser.

\begin{lemma}\label{lem:x1}
Under the value-domain assumptions of \Cref{sec:policies}, for
$X = \{v_1,\dots,v_k\}$ and $C_X = (u := \langle v_1,\dots,v_k\rangle)$, the
least sound bound for $\{u\}$ is $\bigotimes_{v\in X}\sigma(v) =
\alpha_\sigma(\hat X)$.
\end{lemma}

\begin{theorem}[minimality of the read-out]\label{thm:x}
Every join preserving read-out $\beta : \Uset\to M$ that is sound for
$\sigma$ satisfies $\beta \sqsupseteq \alpha_\sigma$ pointwise, so
$\alpha_\sigma$ is the least sound join preserving read-out.
\end{theorem}
\begin{proof}
The observation in \Cref{lem:x1} is $\langle g_{v_1},\dots,g_{v_k}\rangle$,
whose kernel is $\bigsqcup_{v\in X}\ker g_v$ by injectivity of the pairing;
as $g$ ranges over legal families each $\ker g_v$ ranges over all of
$\sigma(v)$, since every $P\in\sigma(v)$ is the kernel of a total mediator valued in the encoding of $D/P$. The least sound bound is the tiling closure of the realised kernels,
which is $\bigotimes_{v\in X}\sigma(v)$, proving \Cref{lem:x1}. The universal
type of $u$ after $C_X$ is $\hat X$, so soundness of $\beta$ forces
$\beta(\hat X)\sqsupseteq\alpha_\sigma(\hat X)$; every $\mathbb{X}$ is
$\lub_{X\in\mathbb{X}}\hat X$ and both maps preserve joins.
\end{proof}

\begin{corollary}[the locale object is strictly coarser]\label{cor:loc}
The read-out induced by the free locale $\Uloc$, namely
$X \mapsto \bigotimes_{v\in X}\sigma(v)$ precomposed with $u$, is monotone, sound, and pointwise above $\alpha_\sigma$. It is strictly
above as soon as some $\Delta_C(w)$ has two incomparable maximal sets. For
$\textsf{if}\;p(x)\;(w := g(y))\;(w := h(z))$ with classical labels it gives
$\dc\!\sim_{xyz}$ where $\Uset$ gives $\tc\{\sim_{xy},\sim_{xz}\}$.
\end{corollary}
\begin{proof}
Domination is direct: $\bigotimes_{v\in u\mathbb{X}}\sigma(v)\sqsupseteq\bigotimes_{v\in X}\sigma(v)$ for every $X\in\mathbb{X}$ by affineness, and the join of the right sides is $\alpha_\sigma(\mathbb{X})$. Soundness follows by weakening from \Cref{thm:psnd}. \Cref{thm:x} does not apply to this map: it sends $\dc\{x\}\vee\dc\{y\}$ to $\sigma(x)\ten\sigma(y)$ rather than $\sigma(x)\vee\sigma(y)$, so it does not preserve joins and its domination has to be proved directly, as above. For strictness, take
classical labels and the displayed program: $\Uset$ gives
$\dc\!\sim_{xy}\vee\dc\!\sim_{xz} = \tc\{\sim_{xy},\sim_{xz}\}$ while the
locale gives $\dc\!\sim_{xyz}$. The latter is strictly larger because
$\sim_{xyz}$ is finer than both generators, so it is not below any of their
mixes and hence not a member of $\tc\{\sim_{xy},\sim_{xz}\}$, while every mix
of the generators is coarser than $\sim_{xyz}$.
\end{proof}

So the literal transcription of the lattice-era object is tight on
\Cref{ex:running} but blurs the branch structure, which is the one thing the
quantale's join was introduced to keep. Among the three objects only $\Uset$
records both, and \Cref{thm:nogo} says what that costs.

\begin{figure}[t]
\centering
\begin{tikzpicture}[font=\small]
  \draw[thick,->] (0,-0.15) -- (0,3.75)
       node[above,font=\scriptsize]{more information};
  \fill (0,0.45) circle (1.7pt);
  \node[right=5pt,font=\scriptsize,align=left] at (0,0.45)
       {$\alpha_\sigma(\Jset_W)$ \;\emph{floor}\\
        least sound join-preserving\\ read-out (\Cref{thm:x}), attained by \Cref{thm:t4}};
  \fill (0,2.05) circle (1.7pt);
  \node[right=5pt,font=\scriptsize,align=left] at (0,2.05)
       {$\aln_\sigma(\Jset^{\mathbb{N}}_W)$ \;\emph{ceiling}\\
        every renaming-capable\\ mechanism (\Cref{thm:nogo})};
  \fill (0,3.25) circle (1.7pt);
  \node[right=5pt,font=\scriptsize,align=left] at (0,3.25)
       {$\top$ \;no guarantee};
  \draw[decorate,decoration={brace,amplitude=4pt}] (-0.12,2.05) -- (-0.12,0.45);
  \node[left=7pt,font=\scriptsize,align=right] at (-0.12,1.25)
       {same-source\\ reuse};
  \draw[dashed] (-0.12,3.25) -- (-0.12,2.05);
  \node[left=7pt,font=\scriptsize,align=right] at (-0.12,2.65)
       {second read,\\ saturated label};
\end{tikzpicture}
\caption{Both ends are pinned. Nothing that supports renaming reaches below the ceiling and no sound join-preserving read-out reaches below the floor; the two coincide
exactly when every label is $\ten$-idempotent (\Cref{thm:dichotomy}). Under a
saturated label the ceiling jumps to $\top$ at the second read of a source
(\Cref{cor:sat}).}
\label{fig:sandwich}
\end{figure}
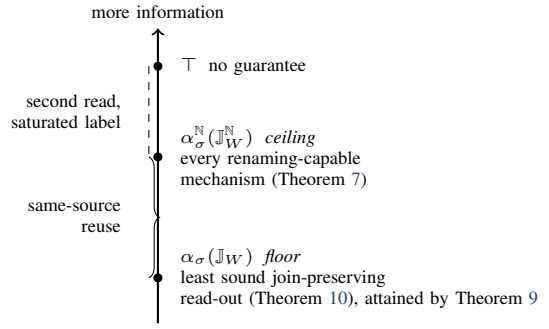

\begin{corollary}[sandwich]\label{cor:sandwich}
For every $C$, $W$ and affine $\sigma$,
$\alpha_\sigma(\Jset_W)\sqsubseteq\aln_\sigma(\Jset^{\mathbb{N}}_W)$, the left side being the floor of sound join-preserving read-outs and the right side the ceiling of the renaming paradigm, with equality iff the labels are idempotent.
\end{corollary}
\begin{proof}
$\alpha_\sigma(\Jset_W) =
\aln_\sigma(s(q(\Jset^{\mathbb{N}}_W)))\sqsubseteq
\aln_\sigma(\Jset^{\mathbb{N}}_W)$ using \Cref{prop:decomp}, the fact that
$q$ is a homomorphism and so commutes with the algorithmic typing
(\Cref{prop:a1}), and the counit $s\circ q\sqsubseteq\mathrm{id}$. The two
readings of the sides are \Cref{thm:x} and \Cref{thm:nogo}, and the equality
condition is \Cref{thm:dichotomy}.
\end{proof}

\begin{proposition}[decidability]\label{prop:dec}
For finite $\Var$: $\Uset(\Var)$ is finite and its elements are represented
by the antichain of their maximal sets, so universal typing terminates and
compliance against a finite $\QoI(D)$ target is decidable. $\Umul(\Var)$ is
infinite and has no ascending chain condition, but downsets of
$\mathbb{N}^{\Var}$ are finite unions of ideals by Dickson's lemma, and by
\Cref{prop:order} only multiplicities below the largest label order affect
the specialised bound, so a finite truncation decides the multiset side as well, made precise next.
\end{proposition}
\begin{proof}
$\Uset(\Var)$ is a finite lattice, so the fixed point of \Cref{sec:alg} is reached by a finite Kleene chain and each type is stored
as the antichain of its maximal sets. Membership and inclusion in a finite
$\QoI(D)$ are decidable, whence compliance is. For $\Umul(\Var)$, downward
closed subsets of $\mathbb{N}^{\Var}$ are finite unions of ideals
$\prod_v[0,a_v]$ with $a_v\in\mathbb{N}\cup\{\infty\}$, because their
complements are upward closed and finitely generated by Dickson's lemma; by
\Cref{prop:order} the specialised bound depends only on multiplicities
truncated at $\max_v\ord\sigma(v)$, so it may be computed in the finite
truncation.
\end{proof}

\begin{proposition}[finite truncation]\label{prop:trunc}
Let $N \geq \max_v\ord\sigma(v)$, let $\mathbb{N}_{\leq N}$ be
$\{0,\dots,N\}$ under the capped sum $\min(a{+}b,N)$, and let
$\Umul^{\leq N}$ be the downset quantale of $(\mathbb{N}_{\leq N})^{\Var}$.
Capping multiplicities extends to a surjective quantale homomorphism
$c_N:\Umul\to\Umul^{\leq N}$, the algorithmic systems commute with it, and
$\aln_\sigma$ factors through it, so
$\aln_\sigma(\Jset^{\mathbb{N}}_W)$ is computed by running \Cref{fig:rules}
in the finite quantale $\Umul^{\leq N}$.
\end{proposition}
\begin{proof}
Capping is a surjective monotone monoid homomorphism, since
$\min(a{+}b,N) = \min(\min(a,N){+}\min(b,N),\,N)$, so its downset extension
preserves joins, tensor and unit. Commutation with $\Alg$ is the transport
clause of \Cref{prop:a1} for homomorphisms, and the factorisation is
$\sigma(v)^{k} = \sigma(v)^{\min(k,N)}$ for $N\geq\ord\sigma(v)$.
\end{proof}

Moving from $\Uloc$ to $\Uset$ turns a type from a set of variables into an
antichain of them. That is the price of disjunction rather than of our
design: the target's join no longer collapses, so a type must be able to say
``this run depended on $X_1$, that run on $X_2$'', and collapsing the
antichain is exactly the strictly coarser locale of \Cref{cor:loc}. On
classical labels the antichains are singletons and the lattice-era cost
returns.

Concretely, a type in $\Uset$ is stored as the antichain of maximal sets, of
which there are at most $\binom{n}{\lfloor n/2\rfloor}$ for $n = |\Var|$, and
the algorithmic system computes with unions and pointwise unions of
antichains followed by removal of dominated elements. Each loop reaches its
invariant in at most $|\Uset|\cdot n$ steps because the lattice is finite and
the iteration is monotone, and a single derivation serves every policy, so
this cost is paid once per program. In $\Umul$ a type is a finite union of
ideals $\prod_v[0,a_v]$ with $a_v\in\mathbb{N}\cup\{\infty\}$, and by
\Cref{prop:order} the entries may be capped at the largest label order, which
for the ethical wall is two; the multiset side is therefore computable as
well, and is what one would compute if one wanted the bound that a
renaming-based mechanism would certify.

\subsection{What the repair keeps}\label{sec:practice}

Monotone renaming is a statement about \emph{derivations}: it transports a
proof from one policy to another. The property that made a universal object
attractive for enforcement is weaker than that. \cite{vanDelft15} names it:
type inference is independent of the policy to be enforced, so one generic
typing verifies many policies of a program. That property survives intact
here. To discharge a compound policy $\bigwedge_i(W_i,\mathit{Pol}_i)$
against a program $C$:
\begin{enumerate}[leftmargin=1.4em,itemsep=1pt]
\item compute $\Delta_C$ once in $\Uset$ by the algorithmic system of \Cref{sec:alg}, which mentions no labelling and terminates by \Cref{prop:dec};
\item for each conjunct, form $\Jset_{W_i}$, evaluate $\alpha_\sigma$ on its
maximal sets, and test the result against $\mathit{Pol}_i$ in $\QoI(D)$.
\end{enumerate}
Step 1 is policy-independent and is performed once per program. Step 2 is a
join of tensors over an antichain. Changing the labelling, adding a conjunct
or retargeting the analysis at a different quantale re-runs only step 2. What
deferring specialisation costs is not this reuse but the ability to present
the outcome as a derivation \emph{inside} the target system, and by
\Cref{thm:nogo} that ability is exactly what would cap the precision: a
certificate that can be replayed as an $M$-derivation is a certificate no
better than the multiset bound.

\subsection{A verification, end to end}\label{sec:worked2}

Take a report generator over a public flag $p$, a client summary $s$ under
the wall $\mathbb{S}$ of \Cref{ex:wall}, and a second source $t$ carrying a
classical label $\sigma(t) = \dc Q$:
\[
C \;=\; (\textsf{if}\;p\;\textsf{then}\;a := f(s)\;\textsf{else}\;
          a := g(t));\;\; b := s .
\]
The policy has two conjuncts: the header alone must stay within one wall,
$\mathit{Pol}_1 = \mathbb{S}\vee\dc Q$ for $W_1 = \{a\}$, and the pair must
stay within the wall combined with $t$, $\mathit{Pol}_2 =
\mathbb{S}\ten\dc Q$ for $W_2 = \{a,b\}$.

Universal typing runs once. With $\Delta_0(v) = \hat v$ and context $1$, the
conditional gives $\Delta_C(a) = \hat p\ten(\hat s\vee\hat t) =
\dc\{p,s\}\cup\dc\{p,t\}$, and the assignment gives $\Delta_C(b) = \hat s$.
Reading out, $\Jset_{W_1} = \Delta_C(a)$ and
\[
\Jset_{W_2} \;=\; \Delta_C(a)\ten\hat s
   \;=\; \dc\{p,s\}\cup\dc\{p,s,t\} \;=\; \dc\{p,s,t\},
\]
because $\{p,s\}\cup\{s\} = \{p,s\}$ on the first branch and
$\{p,t\}\cup\{s\} = \{p,s,t\}$ on the second. Specialising with
$\sigma(p) = 1$: the first conjunct gets $\mathbb{S}\vee\dc Q$ and the second
$\mathbb{S}\ten\dc Q$. Both hold, so the program is certified.

The multiset object runs the same program to
$\Delta^{\mathbb{N}}_C(a) = \dc(2e_p{+}e_s)\cup\dc(2e_p{+}e_t)$ and
$\Delta^{\mathbb{N}}_C(b) = \dc e_s$, so
$\Jset^{\mathbb{N}}_{W_2} = \dc(2e_p{+}2e_s)\cup\dc(2e_p{+}e_s{+}e_t)$ and
the second conjunct is specialised to
$\mathbb{S}^{2}\vee(\mathbb{S}\ten\dc Q) = \top$. The doubled $p$ costs
nothing because $\sigma(p) = 1$, and the doubled $s$ costs everything. By
\Cref{thm:nogo} that is what every renaming-capable mechanism reports, and by
\Cref{prop:local} the single reused source $s$ is the whole reason.

Semantically the certified bound is right. On a run where $p$ holds, the pair
is $(f(s),s)$ and reveals no more than $s$ does; on a run where it does not,
the pair is $(g(t),s)$ and reveals no more than $t$ and $s$ together. The
observation is a mix of the two, which is a member of $\mathbb{S}\ten\dc Q$.
The gap between the two objects is one branch's worth of double counting on a
source the program reads twice.

\section{Three patterns}\label{sec:patterns}

The dichotomy is not a blanket verdict on disjunctive enforcement. It
separates three uses of a walled source, and the separation is the practical
content of the paper. Take $\sigma(s) = \sigma(s_1) = \sigma(s_2) =
\mathbb{S}$, the wall of \Cref{ex:wall}, and a policy asking that the
observation stay within $\mathbb{S}$.

\emph{One source, one use.} A report reads $s$ once. The universal type of
the output is $\hat s$ and both bounds are $\mathbb{S}$, so the policy is
certified by any mechanism in the field. Nothing in this paper is needed, and
nothing in this paper takes it away.

\emph{One source, two uses.} \Cref{ex:running}. The universal type is
$\hat s\ten\hat s = \hat s$ and the set bound is $\mathbb{S}$, which is the
truth, while the multiset bound is $\mathbb{S}\ten\mathbb{S} = \top$. The
program satisfies the policy and only the deferred read-out certifies it
(\Cref{thm:fr}).

\emph{Two sources, one use each.} A report reads $s_1$ and $s_2$, two values
each separately within the wall. Here the universal type is $\hat s_1\ten
\hat s_2 = \dc\{s_1,s_2\}$ and \emph{both} bounds are
$\mathbb{S}\ten\mathbb{S} = \top$, and that is correct: the two mediators may
honour different disjuncts, one revealing $r_1$ and the other $r_2$, so the
pair reveals the secret and the wall is genuinely breached.

A programmer who knows the second pattern is being rejected cannot rewrite
around it. Introducing a temporary, $x := s;\; a := x;\; b := x$, leaves the
$M$-system exactly where it was: soundness forces $\Gamma'(a)$ and
$\Gamma'(b)$ to contain every member of $\mathbb{S}$, and the read-out is
their tensor, so \Cref{prop:ceiling} applies verbatim. The obstruction lives
in the type domain rather than in the program text, which is why the response
has to be a different object rather than a different coding style.

The third pattern is the one the tensor was introduced for, and the set
object does not weaken it. What the set object adds is the ability to tell the
second pattern from the third, which the multiset object cannot: it sees
$\dc(2e_s)$ and $\dc(e_{s_1}{+}e_{s_2})$ and specialises both to $\top$.
Reading a walled value twice and reading two walled values are different
events, and the distinction survives exactly one step of specialisation, the
step that \Cref{thm:nogo} shows a renaming-capable mechanism cannot take.

\section{Related work}\label{sec:related}

\emph{Flow-sensitive types and principal typings.} Lattice-based information flow begins with certification~\cite{DenningDenning77} over Denning's lattice model~\cite{Denning76} and its type-theoretic formulation~\cite{Volpano96}; see~\cite{SabelfeldMyers03} for the classical survey, and~\cite{ClarkHankinHunt02} for dependency analysis of the kind this paper types. Principal typings in general are surveyed in~\cite{Wells02}. The flow-sensitive family, its universal object and internal completeness are from~\cite{HuntSands06}, with polynomial-time inference of the principal typing in~\cite{HuntSands11}, and the relation to Hoare-style independence logic~\cite{AmtoftBanerjee04} is established there. Dynamic policies received
a principal treatment in the same style~\cite{vanDelft15}, which isolates
policy-independent inference as the feature worth preserving, and that is the
feature \Cref{sec:practice} shows survives the repair. All of this assumes
the combination operator is a join, which is the hypothesis removed here.

\emph{Quantale of information.} The lattice of information~\cite{LandauerRedmond93} has semantic refinements such as the demonic lattice~\cite{Morgan17} and, on the analysis side, per-based models of dependency~\cite{SabelfeldSands01}. The quantale's semantic model, the composition lemmas and the ethical-wall examples are from~\cite{HuntSands21}, which leaves the type system layer open. \Cref{fig:rules} is anchored to their lemmas rule by
rule. Their abstract domain map is, in our notation, the specialisation
$\alpha_{\sigma_0}$ at the all-principal labelling, so the positive side of
\Cref{thm:dichotomy} has in effect been instantiated once by the authors,
while the disjunctive side has not been approached. Their closing section
also asks for relational analyses and suggests path sensitivity as a route,
and \Cref{rem:scope} places our results with respect to that programme.

\emph{Disjunctive enforcement.} Disjunctive policies have been enforced for a
database query language against a purpose-built security
condition~\cite{Ahmadian24}. The homomorphism from queries into the quantale
used there is a single concrete map rather than a family indexed by
labellings, and principality, specialisation and completeness are not
addressed. By \Cref{thm:nogo} any such mechanism that supports renaming is
subject to the same ceiling.

\emph{Quantales.} Free quantales and the downset construction with Day
convolution are standard~\cite{Rosenthal90} and are used as given. The
contribution is that the side condition of the universal property is the
dichotomy condition of the type system.

\emph{Abstract interpretation.} Completeness of abstract interpretations is a classical subject~\cite{GRS00}, and completeness of abstract domains for dependency in particular has been studied through abstract non-interference~\cite{GiacobazziMastroeni04}. The question there is which
abstractions are complete for a fixed semantics. Here it is which universal
object an entire family of type systems specialises from, and the answer is a
pair of objects rather than one. Disjunctive completion and the powerset
operator on abstractions~\cite{CousotCousot79,FileRanzato99} are the
classical counterpart of the step from $\Uloc$ to $\Uset$;
\Cref{rem:disjcomp} explains why the step only pays once the concrete side
is a quantale.

\section{Conclusion}\label{sec:conclusion}

Disjunctive policies were given a semantics before they were given an
enforcement mechanism, and building the mechanism turns out to force a choice
that the lattice era never had to make. The universal object that made
flow-sensitive typing reusable across policies splits: the free commutative
quantale on the variables keeps the entire mechanism and loses same-source
correlation, the free object with idempotent generators keeps precision and
loses renaming, and the lattice-era object loses branch disjunction. No
object with a renaming-capable specialisation improves on the first, so for
the ethical-wall and secret-sharing labels a second read of a source leaves
nothing to certify, and programs that respect the wall are rejected.
Deferring specialisation to the judgement level recovers the missing precision, through the least sound join-preserving read-out.

Two directions follow. The characterisation of idempotents, and with it the
reading of the dichotomy as classical against disjunctive labels, is stated
over a finite set of secrets, and the infinite case is open. And the
relational hierarchy that~\cite[\S VIII]{HuntSands21} calls for begins where
\Cref{sec:repair} stops: keeping shared sources visible until the read-out
recovers the correlation between two uses of one variable, while correlating
two uses of one \emph{branch} needs path sensitivity and a mechanism outside
the shape studied here.

\section*{Acknowledgment on the use of AI}

The authors used an AI assistant (Anthropic's Claude) in preparing this work:
in exploring and cross-checking the algebraic derivations underlying
\Cref{sec:objects,sec:survive,sec:tight,sec:nogo,sec:repair}, and in drafting
the text throughout. Every definition, theorem statement and proof was
reviewed by the authors, who are responsible for the content.

\appendices

\section{Soundness}\label{app:t0}

Throughout, $g$ is a mediator family with common domain
$\mathbb{D}\subseteq D$, and $h : \mathbb{P}\imp\mathrm{Id}$ means that some
$P\in\mathbb{P}$, an equivalence relation on all of $D$, satisfies
$h(d) = h(d')$ whenever $d\,P\,d'$ and both sides are defined. Witnesses
restrict: if $P$ works on $\mathbb{D}$ it works on any subset.

\begin{lemma}[joint observation]\label{lem:b1}
If $h_i : \mathbb{P}_i\imp\mathrm{Id}$ for $i$ in a finite set, then
$\langle h_i\rangle_i : \bigotimes_i\mathbb{P}_i\imp\mathrm{Id}$.
\end{lemma}
\begin{proof}
Take witnesses $P_i\in\mathbb{P}_i$ and put $P = \bigsqcup_i P_i$. If
$d\,P\,d'$ then $d\,P_i\,d'$ for each $i$, so all components agree, and
$P$ is a generator of $\bigotimes_i\mathbb{P}_i$.
\end{proof}

\begin{lemma}[coarsening]\label{lem:b2}
If $g_v : \Gamma(v)\imp\mathrm{Id}$ for $v\in V$ and $h = k\circ\langle
g_v\rangle_{v\in V}$ for any $k$, then
$h : \bigotimes_{v\in V}\Gamma(v)\imp\mathrm{Id}$.
\end{lemma}
\begin{proof}
The witness of \Cref{lem:b1} works, since $k$ maps equals to equals.
\end{proof}

\begin{lemma}[conditional]\label{lem:b3}
Let $b$ be a total function on $D$ with
$b : \mathbb{P}\imp\mathrm{Id}$, let $D_1,D_2$ be the two fibres of $b$
after composing with a two-valued test, and let $k_i$ be partial functions
with $k_i : \mathbb{R}_i\imp\mathrm{Id}$ on $D_i$. Then the function equal to
$k_1$ on $D_1$ and $k_2$ on $D_2$ satisfies
$\mathbb{P}\ten(\mathbb{R}_1\vee\mathbb{R}_2)\imp\mathrm{Id}$.
\end{lemma}
\begin{proof}
Let $Q\in\mathbb{P}$ witness $b$ and $R_i\in\mathbb{R}_i$ witness $k_i$.
Since $b$ is total and $Q$-invariant, $D_1$ and $D_2$ are unions of
$Q$-cells, hence unions of $(Q\sqcup R_i)$-cells. Let $R$ be the equivalence
whose cells are the $(Q\sqcup R_1)$-cells contained in $D_1$ together with
the $(Q\sqcup R_2)$-cells contained in $D_2$; this is a partition of $D$ by
the previous sentence. Each of its cells is a cell of $Q\sqcup R_i$, and
$Q\sqcup R_i\in\mathbb{P}\ten\mathbb{R}_i\sqsubseteq
\mathbb{P}\ten(\mathbb{R}_1\vee\mathbb{R}_2)$, so $R$ lies in the mix of
members of the latter, which is tiling closed. If $d\,R\,d'$ then $d,d'$ lie
in the same fibre $D_i$ and in one $R_i$-cell, so the outputs agree where
defined.
\end{proof}

This is~\cite[Lem.~6]{HuntSands21} in $\imp$ form for partial branches; the
tagging step of their proof is replaced by the observation that the fibres of
a total guard are saturated by the guard's own witness, which is what makes
the mix legitimate.

\begin{lemma}[context bound]\label{lem:ctx}
If $\derM{p}{\Gamma}{C}{\Gamma'}$ then $\Gamma'(x)\sqsupseteq p$ for every
$x\in\asgn(C)$.
\end{lemma}
\begin{proof}
Induction. \textsc{Assign} gives $p\ten t\sqsupseteq p$ by affineness;
\textsc{Seq} splits on which component assigns $x$; \textsc{If} gives
$t\ten(\Gamma_1'\vee\Gamma_2')(x)\sqsupseteq t\ten p\ten t\sqsupseteq p$;
\textsc{While} gives $t^{*}\ten\Gamma(x)\sqsupseteq\Gamma(x)\sqsupseteq
p\ten t\sqsupseteq p$; \textsc{Sub} weakens on both sides.
\end{proof}

\begin{proof}[Proof of \Cref{thm:t0}]
Induction on the derivation. Condition (i) is immediate from
\Cref{lem:ctx} in every case, so we treat (ii). Fix legal $g$ and $W$.

\textsc{Skip} is trivial and \textsc{Sub} is monotonicity of
$\imp$ in both arguments together with \Cref{lem:ctx}.

\textsc{Assign} $x := E$. For $w\in W\setminus\{x\}$ the mediator is
unchanged. For $x$, $\pi_x\circ\sem{C}\circ\hat g = \sem{E}\circ\langle
g_v\rangle_{v\in\fv(E)}$, which by \Cref{lem:b2} satisfies
$t\imp\mathrm{Id}$, and $t\sqsubseteq p\ten t = \Gamma'(x)$ by affineness.
Assemble over $W$ with \Cref{lem:b1}.

\textsc{Seq}. By the first hypothesis the family
$g'_v = \pi_v\circ\sem{C_1}\circ\hat g$ is legal for $\Gamma''$; it is the
mediator family the second hypothesis needs, and
$\sem{C_1;C_2}\circ\hat g = \sem{C_2}\circ\hat g'$. Partiality is absorbed
because $\imp$ only constrains pairs on which both sides are defined. This is
the chain rule~\cite[Lem.~5]{HuntSands21}, obtained here by the definition of
the semantics rather than by a separate argument.

\textsc{If}. Let $b$ be the guard evaluated through $\hat g$; by
\Cref{lem:b2}, $b : t\imp\mathrm{Id}$. The two hypotheses, with $g$
restricted to the corresponding fibre, give
$\pi_{W_a}\circ\sem{C_i}\circ\hat g : \bigotimes_{W_a}\Gamma_i'\imp
\mathrm{Id}$ where $W_a = W\cap(\asgn(C_1)\cup\asgn(C_2))$. \Cref{lem:b3}
yields $t\ten\bigl(\bigotimes_{W_a}\Gamma_1'\vee
\bigotimes_{W_a}\Gamma_2'\bigr)\imp\mathrm{Id}$, and this is below
$\bigotimes_{w\in W_a}\bigl(t\ten(\Gamma_1'\vee\Gamma_2')(w)\bigr)$: expanding
the right side by distributivity gives, among others, the terms in which the
same branch is chosen at every component, each carrying $|W_a|$ copies of
$t$, and $t^{|W_a|}\sqsupseteq t$ by affineness. Weakening
(\cite[Prop.~1]{HuntSands21}) finishes; unassigned components are handled by
their own mediators and \Cref{lem:b1}.

\textsc{While}. Write $s_0 = \hat g$ and $s_{k+1} = \sem{C}\circ s_k$
where defined, $b_k$ for the guard at step $k$, $\mathbb{D}'$ for the set of
$d$ at which the loop terminates, $n(d)$ for the number of iterations and $K = \{k \mid \exists d\in\mathbb{D}',\ k\leq n(d)\}$, which is finite because $D$ is. By induction on $k$,
using the hypothesis with the mediator family at step $k$ restricted to the
set where the guard holds, every family $\langle\pi_v\circ s_k\rangle_v$ is
legal for $\Gamma$; fix witnesses $P^{(k)}_v\in\Gamma(v)$ and put
$Q_k = \bigsqcup_{v\in\fv(E)}P^{(k)}_v\in t$ and
$R_k = \bigsqcup_{w\in W}P^{(k)}_w\in\bigotimes_W\Gamma(w)$.

Let $T = \bigsqcup_{k\in K}Q_k$, a member of $t^{|K|}$ and hence of $t^{*}$.
If $d,d'\in\mathbb{D}'$ and $d\,T\,d'$ then $n(d) = n(d')$: assuming
$n(d)\leq n(d')$, both runs reach every step $j\leq n(d)$, so $Q_j$ applies
and $b_j(d) = b_j(d')$ for all such $j$; at $j = n(d)$ the guard is false for
$d$, hence for $d'$. Consequently $n$ is constant on
$\gamma\cap\mathbb{D}'$ for every $T$-cell $\gamma$.

If $W\cap\asgn(C) = \emptyset$ the observation is the initial one and
$\bigsqcup_{w\in W}\ker g_w$ is a witness. Otherwise fix
$w_0\in W\cap\asgn(C)$ and define $R$ by taking, for each $T$-cell $\gamma$
disjoint from $\mathbb{D}'$, the cell $\gamma$ itself, and for each other
$T$-cell $\gamma$, with $k$ the constant value of $n$ on
$\gamma\cap\mathbb{D}'$, the $(T\sqcup R_k)$-cells contained in $\gamma$. Both
kinds are cells of members of $\bigotimes_W\Gamma^{\dagger}(w)$: the first of
$T$ itself, which lies in $t^{*}\sqsubseteq\Gamma^{\dagger}(w_0)$, and the
second of $(T\sqcup P^{(k)}_{w_0})\sqcup\bigsqcup_{w\neq w_0}P^{(k)}_w$,
whose first component lies in $t^{*}\ten\Gamma(w_0) =
\Gamma^{\dagger}(w_0)$. Hence $R\in\bigotimes_W\Gamma^{\dagger}(w)$. If
$d\,R\,d'$ with both terminating then $n(d) = n(d') = k$ and $d\,R_k\,d'$, so
the final observations agree.
\end{proof}

The only step that consumes more than one copy of the guard type is the
identification of the iteration count, and it consumes finitely many; $t^{*}$
buys all of them at once. In a locale the copies collapse and the rule
of~\cite{HuntSands06} reappears.

\section{The algorithmic system}\label{app:a1}

The algorithmic system is displayed in \Cref{sec:system}.

\begin{proof}[Proof of \Cref{prop:a1}]
\emph{Monotonicity.} Induction on $C$. The base and composite cases are
monotonicity of $\ten$, $\vee$ and composition. For the loop, the operator
$F_{p,\Gamma}(\Theta) = \Gamma\vee\Alg^{C}(p\ten t_E(\Theta),\Theta)$ is
monotone in all three arguments by the inductive hypothesis, so
$F_{p,\Gamma}\sqsubseteq F_{\bar p,\bar\Gamma}$ pointwise when
$p\sqsubseteq\bar p$ and $\Gamma\sqsubseteq\bar\Gamma$; then
$\mathrm{lfp}\,F_{\bar p,\bar\Gamma}$ is a pre-fixed point of $F_{p,\Gamma}$
and Knaster--Tarski gives
$\mathrm{lfp}\,F_{p,\Gamma}\sqsubseteq\mathrm{lfp}\,F_{\bar p,\bar\Gamma}$.
Monotonicity of $(\cdot)^{*}$ and $\ten$ finishes.

\emph{Derivability.} Induction on $C$; only the loop is not immediate. By the
fixed-point property, $\Alg^{C}(p\ten t_E(\Gamma^{*}),\Gamma^{*})
\sqsubseteq\Gamma^{*}$, so the inductive hypothesis and \textsc{Sub} give
$\derM{p\ten t_E(\Gamma^{*})}{\Gamma^{*}}{C}{\Gamma^{*}}$, which is the
premise of \textsc{While}; its conclusion has precondition $\Gamma^{*}
\sqsupseteq\Gamma$, and \textsc{Sub} tightens it to $\Gamma$.

\emph{Minimality.} Induction on the derivation of
$\derM{p}{\Gamma}{C}{\Gamma'}$. \textsc{Sub} uses monotonicity;
\textsc{Seq} composes the two hypotheses and uses monotonicity in the second
argument; \textsc{If} uses monotonicity of $\ten$ and $\vee$. For
\textsc{While} the rule's precondition is the invariant $\Gamma$ itself and
the hypothesis gives $\Alg^{C}(p\ten t_E(\Gamma),\Gamma)\sqsubseteq\Gamma$,
so $F_{p,\Gamma}(\Gamma)\sqsubseteq\Gamma$ and hence
$\Gamma^{*}\sqsubseteq\Gamma$; monotonicity of $t_E$, $(\cdot)^{*}$ and
$\ten$ then gives $\Gamma^{*\dagger}\sqsubseteq\Gamma^{\dagger}$.

\emph{Transport.} If $f$ is a homomorphism it commutes with $\ten$, $\vee$
and $(\cdot)^{*}$, the last because
$f(\lub_n t^{n}) = \lub_n f(t)^{n}$; since $f$ preserves all joins it
commutes with the least fixed point, by transfinite induction on its
iteration. If $f$ only satisfies the premises of \Cref{lem:qren}, then
\Cref{lem:qren} makes $f\circ\Alg^{C}_M(p,\Gamma)$ a derivable postcondition
in $M'$ from $(fp,f\circ\Gamma)$, and minimality in $M'$ gives the stated
inequality.
\end{proof}

\end{document}